\documentclass[a4paper,11pt]{article}

\usepackage{amssymb,amsmath}
\usepackage{amsthm}
\usepackage{mathtools}
\usepackage{algorithm}
\usepackage{url}
\usepackage{hyperref}
\usepackage[capitalize]{cleveref}
\usepackage{color}

\usepackage[left=27mm,right=27mm,top=27mm,bottom=27mm]{geometry}

\DeclareMathOperator{\size}{size}
\DeclareMathOperator{\val}{val}
\DeclareMathOperator{\dist}{dist}
\DeclareMathOperator{\MF}{MF}
\DeclareMathOperator{\opt}{opt}

\newtheorem{theorem}{Theorem}[section]
\newtheorem{lemma}[theorem]{Lemma}
\newtheorem{proposition}[theorem]{Proposition}
\newtheorem{corollary}[theorem]{Corollary}
\newtheorem*{claim}{Claim}
\theoremstyle{definition}

\crefname{equation}{}{}

\title{Node-Connectivity Terminal Backup,\\
Separately-Capacitated Multiflow, and Discrete Convexity}
\author{Hiroshi HIRAI and Motoki IKEDA\\
Department of Mathematical Informatics,\\
Graduate School of Information Science and Technology,\\
The University of Tokyo, Tokyo, 113-8656, Japan\\
\texttt{\normalsize \{hirai,motoki\_ikeda\}@mist.i.u-tokyo.ac.jp}}

\begin{document}

\maketitle

\begin{abstract}
The \emph{terminal backup problems} (Anshelevich and Karagiozova (2011)) form
a class of network design problems:
Given an undirected graph with a requirement on terminals,
the goal is to find a minimum cost subgraph satisfying the connectivity requirement.
The \emph{node-connectivity terminal backup problem} requires
a terminal to connect other terminals with a number of
node-disjoint paths.
This problem is not known whether is NP-hard or tractable.
Fukunaga (2016) gave a $4/3$-approximation algorithm
based on LP-rounding scheme using a general LP-solver.

In this paper, we develop a combinatorial algorithm for the relaxed LP
to find a half-integral optimal solution in
$O(m\log (nUA)\cdot \MF(kn,m+k^2n))$ time,
where $n$ is the number of nodes,
 $m$ is the number of edges, $k$ is the number of terminals,
$A$ is the maximum edge-cost, $U$ is the maximum edge-capacity,
and $\MF(n',m')$ is the time complexity of a max-flow algorithm
in a network with $n'$ nodes and $m'$ edges.
The algorithm implies that the $4/3$-approximation algorithm
for the node-connectivity terminal backup problem
is also efficiently implemented.
For the design of algorithm, we explore a connection between
the node-connectivity terminal backup problem and
a new type of a multiflow, called a \emph{separately-capacitated multiflow}.
We show a min-max theorem which extends Lov\'{a}sz--Cherkassky theorem
to the node-capacity setting.
Our results build on discrete convexity in
the node-connectivity terminal backup problem.
\end{abstract}


\section{Introduction}
Network design problems are central problems in combinatorial optimization. 
A large number of basic combinatorial optimization problems  
are network design problems. Examples are spanning tree, matching, TSP, and Steiner networks.
They admit a typical formulation of a network design problem: Find a minimum-cost network satisfying given connectivity requirements.
The present paper addresses a relatively new class of network design problems, called \emph{terminal backup problems}.
The problem is to find a cheapest subnetwork
in which each terminal can send 
a specified amount of flows to other terminals, 
i.e., the data in each terminal can be backed up,
possibly in a distributed manner, in other terminals.

A mathematical formulation of the terminal backup problem is given as follows.
Let $((V,E),S,u,c,a,r)$ be an undirected network,
where $(V,E)$ is a simple undirected graph,
$S\subseteq V\ (\lvert S\rvert\geq 3)$ is a set of \emph{terminals},
$u:E\rightarrow\mathbb{Z}_+$ is a nonnegative edge-capacity function,
$c:V\setminus S\rightarrow\mathbb{Z}_+$ is a nonnegative node-capacity function,
$a:E\rightarrow\mathbb{Z}_+$ is a nonnegative edge-cost function,
and $r:S\rightarrow\mathbb{Z}_+$ is a nonnegative requirement function on terminals.
The goal is to find a feasible edge-capacity function $x$ of minimum cost $\sum_{e \in E}a(e)x(e)$.
Here an edge-capacity function $x$ is said to be \emph{feasible} 
if $0 \leq x \leq u$ and each terminal $s \in S$ has a flow 
from $s$ to $S\setminus\{s\}$,  
an $\{s\}$--$(S\setminus\{s\})$ flow, of total flow-value $r(s)$ 
in the network $((V,E),S,x,c)$ capacitated by the edge-capacity $x$ and the node-capacity $c$.

The original formulation, due to Anshelevich and Karagiozova~\cite{Anshelevich2011Terminal}, 
is uncapacitated (i.e., $u,c$ are infinity), 
requires $x$ to be integer-valued, and assumes $r(s) = 1$ for all $s \in S$.
They showed that an optimal solution can be obtained in polynomial time.  
Bern\'{a}th et al.~\cite{Bernath2015Generalized} 
extended this polynomial time solvability to an arbitrary integer-valued requirement $r$.
For the setting of general edge-capacity (and infinite node-capacity), 
which we call  
the \emph{edge-connectivity terminal backup problem (ETB)}, 
it is unknown whether ETB is NP-hard or tractable.

Fukunaga~\cite{Fukunaga2016Approximating} considered the above setting including both edge-capacity and node-capacity,
which we call the \emph{node-connectivity terminal backup problem (NTB)},
and explored intriguing features of its fractional relaxation. 
The \emph{fractional ETB (FETB)} and \emph{fractional NTB (FNTB)} are
LP-relaxations obtained  from ETB and NTB, respectively,
by relaxing solution $x$ to be real-valued.
Fukunaga showed the half-integrality property of FNTB, that is,
there always exists an optimal solution that is half-integer-valued. 
Based on this property, 
he developed a $4/3$-approximation algorithm for NTB
by rounding a half-integral (extreme) optimal solution.
Moreover, he noticed a useful relationship
between FETB and \emph{multicommodity flow (multiflow)}. 
In fact, a solution of FETB is precisely the edge-support 
of a multiflow consisting of the $r(s)$ amount of 
$\{s\}$--$(S\setminus\{s\})$ flow for each $s \in S$.
This is a consequence of Lov\'asz--Cherkassky theorem~\cite{Cherkassky1977solution,Lovasz1976some} in multiflow theory.
In particular, FETB is equivalent to 
a minimum-cost multiflow problem, 
which is a variant of the one studied
by Karzanov~\cite{Karzanov1979minimum,Karzanov1994Minimum}
and Goldberg and Karzanov~\cite{Goldberg1997Scaling}. 
 
Utilizing this connection, Hirai~\cite{Hirai2015L} developed a 
combinatorial polynomial time algorithm for FETB and 
the corresponding multiflow problem. This algorithm uses 
a max-flow algorithm as a subroutine, and  
brings 
a combinatorial implementation of 
Fukunaga's $4/3$-approximation algorithm for ETB, 
where he used a generic LP-solver (e.g., the ellipsoid method) 
to obtain a half-integral extreme optimal solution.

Our first contribution is an extension of this result to the node-capacitated setting, 
implying that the $4/3$-approximation algorithm for NTB is also efficiently implemented. 
\begin{theorem}
	\label{thm:main}
	A half-integral optimal solution of FNTB can be obtained
	in $O(m\log (nUA)\cdot \MF(kn,m+k^2n))$ time.
\end{theorem}
Here $n:=\lvert V\rvert$, $m:=\lvert E\rvert$,
$k:=\lvert S\rvert$, $U:=\max_{e\in E}u(e)$, and $A:=\max_{e\in E}a(e)$, and 
$\MF(n',m')$ is the time complexity of
an algorithm for solving the max-flow problem
in the network with $n'$ nodes and $m'$ edges.

As in the ETB case, we explore and utilize 
a new connection between NTB and a multiflow problem.
We introduce a new notion of a \emph{free multiflow with separate node-capacity constraints}
 or simply a \emph{separately-capacitated multiflow}.
Instead of the usual node-capacity constraints,
this multiflow should satisfy 
the separate node-capacity constraints: For each terminal $s\in S$ and each node $i\in V$,
the total flow-value of flows
connecting $s$ to the other terminals and flowing into $i$
is at most the node capacity $c(i)$.

Our second contribution is a min-max theorem for separately-capacitated multiflows, 
which extends Lov\'{a}sz--Cherkassky theorem to the node-capacitated setting 
and implies that a solution of FNTB is precisely the edge-support of
a separately-capacitated multiflow. 
This answers Fukunaga's comment: \emph{how the computation should proceed in the node capacitated setting remains elusive}~\cite[p.~799]{Fukunaga2016Approximating}.   

\begin{theorem}
	\label{thm:lc}
	The maximum flow-value of a separately-capacitated multiflow is equal to
	$(1/2)\sum_{s\in S}\nu_s$,
	where $\nu_s$ is the minimum capacity of an $\{s\}$--$(S\setminus\{s\})$ cut.
	Moreover, a half-integral maximum multiflow exists,
    and it can be found in $O(n\cdot\MF(kn,m+k^2 n))$ time.
\end{theorem}
Here, a \emph{$T$--$T'$ cut} is a pair of an edge-subset $F\subseteq E$ and a node-subset $X\subseteq V\setminus (T\cup T')$
such that removing those subsets disconnects $T$ and $T'$, and its capacity 
is defined as $u(F) + c(X)$.

Our algorithm for Theorem~\ref{thm:main} builds on 
the ideas of \emph{Discrete Convex Analysis (DCA) beyond $\mathbb{Z}^n$} ---
a theory of discrete convex functions 
on special graph structures generalizing $\mathbb{Z}^n$ (the grid graph), 
which has been recently differentiated from the original DCA~\cite{Murota2003Discrete} 
and has been successfully applied to algorithm design 
for well-behaved classes of multiflow and related network design 
problems~\cite{Hirai2015L,Hirai2016Discrete,Hirai2018Dual,Hirai2019cost}.
The algorithm in \cite{Hirai2015L} for FETB was indeed designed by this approach: 
Formulate the dual of FETB as a minimization of an \emph{L-convex function} 
on the (Cartesian) product of trees, 
apply the framework of the \emph{steepest descent algorithm (SDA)}, 
and show that it can be implemented by using a max-flow algorithm as a subroutine.

We formulate the dual of FNTB as an optimization problem on 
the product of the spaces of all subtrees of a fixed tree.
We develop a simple cut-descent algorithm for this optimization problem.
Then we prove that this coincides with SDA for an L-convex function defined on 
the graph structure on the space of all subtrees. 
Then the number of descents is estimated by a general theory of SDA, 
and the cost-scaling method is naturally incorporated to derive the time complexity.
\Cref{thm:lc} is obtained as a byproduct of these arguments. 

The rest of this paper is organized as follows.
In \cref{sec:pre}, we introduce notations and basic notions.
In \cref{sec:algo}, we give the combinatorial algorithm for FNTB
and proves \cref{thm:main,thm:lc} without the time complexity results.
The time complexity results are shown in \cref{sec:disc}
using DCA beyond $\mathbb{Z}^n$.

\paragraph*{Related work.}

ETB is a \emph{survivable network design problem (SND)} with a special
skew-supermodular function,
and NTB is a node connectivity version (NSND)
with a special skew-supermodular biset function.
In his influential paper~\cite{Jain2001Factor}, Jain devised the iterative rounding method,
and obtained a 2-approximation algorithm for SND,
provided that an extreme optimal solution of the LP-relaxation of SND
(with modified skew-supermodular functions) is available.
Fleischer, Jain, and Williamson~\cite{Fleischer2006Iterative} and
Cheriyan, Vempala, and Vetta~\cite{Cheriyan2006Network} extended
this iterative rounding 2-approximation algorithm to some classes of NSND.
One of important open problems in the literature is
a design of a combinatorial 2-approximation algorithm for (V)SND
with the skew-supermodular (biset) function associated with connectivity requirements.
One approach is to devise a combinatorial polynomial time algorithm
to find an extreme optimal solution of its LP-relaxation;
the currently known only polynomial time algorithm
is a general LP-solver (e.g., the ellipsoid method).
Our algorithm for FNTB, though it is the LP-relaxation of a very special NSND,
may give an insight on such a research direction.

The notion of a separately-capacitated multiflow, introduced in this paper,
is a new variation of $S$-paths packing.
As seen in \cite[Chapter~73]{Schrijver2003Combinatorial},
$S$-paths packing is one of the well-studied subjects
in combinatorial optimization.
Recent work~\cite{Iwata201801All} developed a fast algorithm
for half-integral \emph{nonzero $S$-paths packing problem
on a group-valued graph} (with unit-capacity).
Our derivation of \cref{thm:lc} is different from flow-manipulation arguments
such as Cherkassky's T-operation~\cite{Cherkassky1977solution} or those in \cite{Iwata201801All}.
This is a future research to establish such an argument for a separately-capacitated multiflow.
Also, exploring an integer version of Theorem~\ref{thm:lc},
an analogue of Mader's theorem~\cite{Mader1978Uber}, is an interesting future direction.

\section{Preliminaries}
\label{sec:pre}

\subsection{Notation}
Let $\mathbb{Z},\mathbb{Z}_+,\mathbb{R},\mathbb{R}_+$ be
the set of integers, nonnegative integers, reals,
and nonnegative reals, respectively.
Let $\mathbb{Z}^*$ and $\mathbb{Z}^*_+$ be the set of half-integers and
nonnegative half-integers, respectively, i.e.,
 $\mathbb{Z}^*:=\mathbb{Z}/2$ and $\mathbb{Z}^*_+:=\mathbb{Z}_+/2$.
Let $\overline{\mathbb{R}}:=\mathbb{R}\cup\{+\infty\}$ and
$\underline{\mathbb{R}}:=\mathbb{R}\cup\{-\infty\}$.
Let denote $(a)^+:=\max\{a,0\}$ for $a\in \mathbb{R}$.
For a finite set $V$,
we often identify a function $f:V\to \mathbb{R}$ with a vector $f=(f_i)_{i\in V}\in \mathbb{R}^V$, by $f_i=f(i)$.
For a subset $U\subseteq V$,
we denote $f(U):=\sum_{i\in U} f(i)$.
For $i\in V$, its characteristic function $\chi_i:V\rightarrow \mathbb{R}$
is defined by $\chi_i(j)=1$ if $j=i$ and $\chi_i(j)=0$ otherwise.
For $U\subseteq V$, let $\chi_U:=\sum_{i\in U}\chi_i$.
We write $V-s:=V\setminus\{s\}$ for $s\in V$.

In this paper, all graphs are simple.
An edge connecting $i$ and $j$ is denoted by $ij$.
For an undirected graph on a node set $V$,
the set of edges connecting $U_1$ and $U_2$ ($U_1,U_2\subseteq V$)
is denoted by $\delta(U_1,U_2)$.
If $U_2=V\setminus U_1$, we simply denote it by $\delta U_1$.
For a singleton $\{i\}$, we write $\delta i$ to denote $\delta\{i\}$.

\subsection{Undirected Circulation}

Let $(U,F)$ be an undirected graph,
and let $\underline{b}:F\rightarrow\underline{\mathbb{R}}$ and
$\overline{b}:F\rightarrow\overline{\mathbb{R}}$
be lower and upper capacity functions
satisfying $\underline{b}(e)\leq \overline{b}(e)$ for each $e\in F$.
The graph $(U,F)$ may contain self-loops (but no multiedges).
The \emph{(undirected) circulation problem} on $((U,F),\underline{b},\overline{b})$
is the problem of finding an edge-weight $y:F\rightarrow\mathbb{R}$
satisfying $\underline{b}(e)\leq y(e)\leq \overline{b}(e)$ for each $e\in F$
and $\sum_{ij\in F} y(ij)=0$ for each $i\in U$.
Such a $y$ is called a \emph{circulation}.

Let $3^U$ denote the set of pairs $(Y,Z)$ of two subsets $Y,Z\subseteq U$
with $Y\cap Z=\emptyset$.
For $(Y,Z)\in 3^U$,
let $\chi_{Y,Z}:=\chi_Y-\chi_Z\in \mathbb{R}^U$.
Let denote $\chi_{Y,Z}(i_1,i_2,\dotsc,i_t)^+:=(\chi_{Y,Z}(\{i_1,i_2,\dotsc,i_t\}))^+$
(possibly $i_t=i_{t'}$ for some distinct $t,t'$).
Define the \emph{cut function} $\kappa:3^U\rightarrow \underline{\mathbb{R}}$ by
\begin{align*}
    \kappa(Y,Z)&:=
        \sum_{ij\in F}
        \chi_{Y,Z}(i,j)^+\underline{b}(ij)
        -\chi_{Z,Y}(i,j)^+\overline{b}(ij)
        \quad ((Y,Z)\in 3^U).
\end{align*}
See \cref{fig:cutfunc} for $\chi_{Y,Z}(i,j)\ (:=\chi_{Y,Z}(\{i,j\}))$.
It is known that the feasibility of the circulation problem is characterized by this cut function.
We can show it by reducing to Hoffman's circulation theorem.
A cut $(Y,Z)\in 3^U$ with $\kappa(Y,Z)>0$ is called \emph{violating},
and is called \emph{maximum violating} if it attains the maximum $\kappa(Y,Z)$ among all violating cuts.

\begin{figure}[t]
\centering
\includegraphics[scale=0.5]{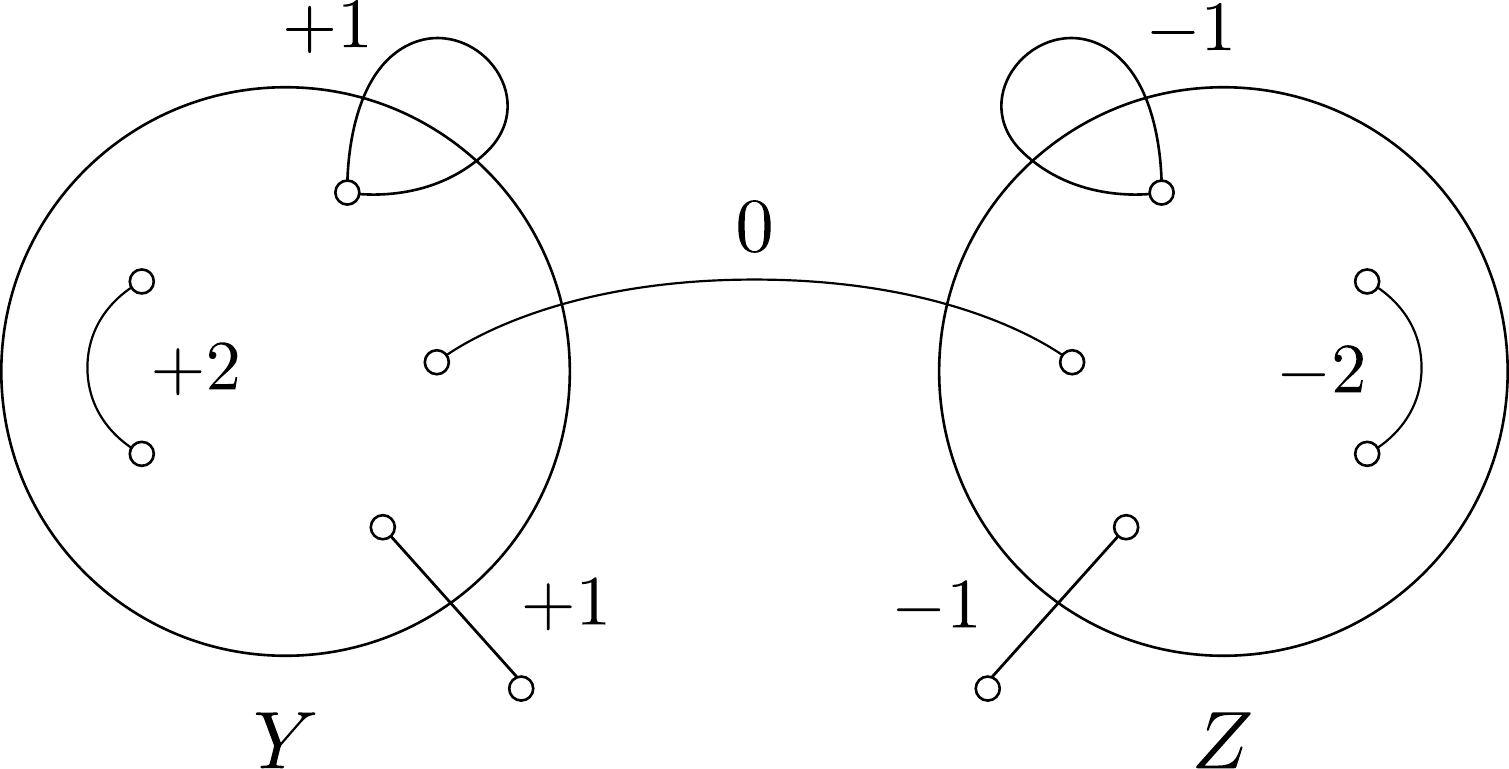}
\caption{$\chi_{Y,Z}(i,j)$.}
\label{fig:cutfunc}
\end{figure}

\begin{lemma}[see, e.g., {\cite[Theorems 2.4, 2.7]{Hirai2019cost}}]
\label{lem:feascirc}
Let $((U,F),\underline{b},\overline{b})$ be an undirected network.
\begin{enumerate}
\renewcommand{\labelenumi}{\textup{(\arabic{enumi})}}
\item The circulation problem is feasible if and only if
$\kappa(Y,Z)\leq 0$ for any $(Y,Z)\in 3^U$.
\item If $\underline{b}$ and $\overline{b}$ are integer-valued,
then there exists a half-integer-valued circulation
$y:E\rightarrow\mathbb{Z}^*$.
\item We can obtain, in $O(\MF(\lvert U\rvert,\lvert F\rvert))$ time,
 a half-integer-valued circulation or a maximum violating cut.
\end{enumerate}
\end{lemma}

\subsection{Fractional NTB}
\label{subsec:lp}

Let $((V,E),S,u,c,a,r)$ be a network.
We assume that $S=\{1,\dotsc,k\}$ in this paper.
A sufficient and necessity condition for the feasibility of NTB
is easily derived from the Menger's theorem as follows.
A \emph{biset} is a pair of node subsets $X,X^+\subseteq V$ with $X\subseteq X^+$.
We write $\hat{X}=(X,X^+)$ for a biset.
Let $\Gamma(\hat{X}):=X^+\setminus X$,
and let $\delta(\hat{X}):=\delta(X,V\setminus X^+)$.
For $s\in S$, let $\mathcal{C}_s$ be a family of bisets defined by
\[
\mathcal{C}_s:=\{(X,X^+)\mid s\in X\subseteq X^+
\subseteq V\setminus (S-s)\}.
\]
Let $\mathcal{C}:=\bigcup_{s\in S}\mathcal{C}_s$.
Then an edge-capacity $x:E\rightarrow\mathbb{Z}_+$ is feasible if and only if
\begin{equation}
\label{cond:ntbfeas}
x(\delta\hat{X})+c(\Gamma(\hat{X}))\geq r(s)\quad (s\in S,\ \hat{X}\in \mathcal{C}_s).
\end{equation}
We assume that $u$ satisfies \cref{cond:ntbfeas} throughout the paper
(otherwise NTB is infeasible).

As in \cite{Fukunaga2016Approximating},
we consider the following LP-relaxation problem FNTB:
\begin{align}
\text{(FNTB)\quad Minimize}\quad & \sum_{e\in E}a(e)x(e)\notag\\
\label{eq:ptb1}\text{subject to}\quad & x(\delta \hat{X})+c(\Gamma(\hat{X}))\geq r(s)\quad (s\in S,\ \hat{X}\in \mathcal{C}_s),\\
\label{eq:ptb2}& 0\leq x(e)\leq u(e)\quad (e\in E).
\end{align}
From the above assumption, the polytope defined by
\cref{eq:ptb1,eq:ptb2} is nonempty.
Also, it is known~\cite[Corollary~3.3]{Fukunaga2016Approximating}
that the polytope is half-integral.
Thus FNTB has a half-integral optimal solution.
In \cite[Lemma~4.4]{Fukunaga2016Approximating},
a general LP solver was used for obtaining such a solution.
The goal of this paper is to develop a combinatorial algorithm for FNTB.

\subsection{Reduction}
\label{subsec:red}

For a technical argument, we convert the original instance
to an equivalent instance satisfying the following condition:
\begin{itemize}
\setlength{\itemindent}{15pt}
\setlength{\labelsep}{10pt}
\item[\textbf{(CP)}] $a(e)>0$ for all $e\in E$.
\end{itemize}
If $Z:=\{e\in E\mid a(e)=0\}$ is nonempty,
we use the following perturbation technique based on \cite{Goldberg1997Scaling,Karzanov1994Minimum}.
Recall that $U$ is the maximum edge capacity.
Define a positive edge-cost $a':E\to \mathbb{Z}_+$ by
$a'(e):=1$ for $e\in Z$ and $a'(e):=(2U\lvert Z\rvert+1)a(e)$ for $e\notin Z$.
Let $x^*$ be a half-integral optimal solution for FNTB under the edge-cost $a'$
(it exists by the half-integrality).
We prove that $x^*$ is also optimal under the original edge-cost $a$.
It suffices to show that
$\sum_{e\in E}a(e)x^*(e)\leq \sum_{e\in E}a(e)x(e)$
for any feasible half-integral edge-capacity $x$.
Observe that
$(2U\lvert Z\rvert+1)(\sum_{e\in E}a(e)x^*(e)
    -\sum_{e\in E}a(e)x(e))
=\sum_{e\in E}a'(e)x^*(e)-\sum_{e\in E}a'(e)x(e)-x^*(Z)+x(Z)
\leq U\lvert Z\rvert$
and thus $\sum_{e\in E}a(e)x^*(e)-\sum_{e\in E}a(e)x(e)\leq
U\lvert Z\rvert/(2U\lvert Z\rvert+1)<1/2$.
By the half-integrality, we obtain $\sum_{e\in E}a(e)x^*(e)-\sum_{e\in E}a(e)x(e)\leq 0$.

The maximum edge-cost of the resulting instance becomes $O(mUA)$.
In the rest of this paper,
we assume that the original instance itself satisfies \textbf{(CP)}
unless we discuss the time complexity result.

\subsection{The Space of Subtrees in a Tree}
\label{subsec:space_of}

We utilize a combinatorial dual problem for FNTB.
In this subsection, we introduce the underlying space of the dual problem,
which we call the \emph{subtree space}.

For each $s\in S$, consider an infinite path graph $P_s$
with one endpoint.
Glue those $k\,(=|S|)$ endpoints, and denote the resulting graph by $\mathbb{T}$.
We also use $P_s$ and $\mathbb{T}$ to represent the node sets of those graphs.
We give length 1/2 for each edge in $\mathbb{T}$.
The glued endpoint is denoted by 0,
and the point in $P_s$ ($s\in S$) having the distance $l$ from $0$
is denoted by $(l,s)$.
We denote the set of all subtrees of $\mathbb{T}$ by $\mathbb{S}=\mathbb{S}(\mathbb{T})$.
If a subtree $T$ does not contain $0$,
then it is contained in some $P_s$; see the right side of \cref{fig:dtb}.
Such a subtree $T$ is said to be of \emph{$s$-type}
and is denoted by $[l,l']_s$,
where $(l,s)$ and $(l',s)$ are the closest and farthest
nodes from $0$ in $T$, respectively.
If a subtree $T$ contains 0,
then it is said to be of \emph{0-type}
and is denoted by a $k$-tuple $[l_1,l_2,\dotsc,l_k]=[l_{s}]_{s\in S}$,
where $(l_{s},s)$ is the node in $T\cap P_s$ farthest from 0 for each $s\in S$.
We also use $[0,l']_s$ for denoting a 0-type subtree
which is contained in $P_s$.
We identify a node on $\mathbb{T}$ with
a subtree consisting of this node only.

In some cases, it is useful
 to denote an $s$-type subtree $[l,l']_s$
 by a $k$-tuple $(-l,-l,\dotsc,l',\dotsc,-l)$
whose $s$-th element is $l'$ and the other elements are all $-l$.
Then any subtree is represented as a $k$-tuple in a unified way.
The subtree space $\mathbb{S}$ can be seen as
the set of vectors $(T_s)_{s\in S}\in (\mathbb{Z}^*)^k$ satisfying
\begin{equation}
\label{eq:subtree}
 T_s \geq 0\ (\forall s\in S) \text{ or } 
 T_s\geq -T_t = - T_{t'} >0\ (\exists s \in S,\ \forall t,t' \in S-s).
\end{equation}

For an $s$-type subtree $T=[l,l']_s\in \mathbb{S}$,
let $\size(T):=l'-l$.
For a 0-type subtree $T=[l_{s}]_{s\in S}\in \mathbb{S}$,
let $\size_s(T):=l_s$ for $s\in S$,
and $\size(T):=\sum_{s=1}^k \size_s(T)$.
For two subtrees $T,T'\in \mathbb{S}$, we denote
the minimum distance between $T$ and $T'$ on $\mathbb{T}$ by $\dist(T,T')$,
i.e., $\dist(T,T'):=\min\{d_{\mathbb{T}}(v,v')\mid v\in T,\ v'\in T'\}$,
where $d_\mathbb{T}$ denotes the distance on $\mathbb{T}$.

\section{Algorithm}
\label{sec:algo}

In this section, we develop a combinatorial algorithm for FNTB with \textbf{(CP)}.
Our algorithm is based on a combinatorial dual problem of FNTB.
In \cref{subsec:dual}, we introduce the duality
and give the outline of the algorithm.
We discuss details in \cref{subsec:opt,subsec:direction},
and give a full description in \cref{subsec:algo}.
This algorithm itself is pseudo-polynomial time,
and in \cref{subsec:scaling}, we combine it with the cost-scaling method
to obtain a weakly polynomial time algorithm.

\subsection{Combinatorial Duality for FNTB}
\label{subsec:dual}

We formulate a dual of FNTB as a problem of
assigning a subtree of $\mathbb{T}$ for each node $i\in V$.
That is, subtrees are viewed as node-potentials.
So we use $p_i$ and $p:V\to \mathbb{S}$ for denoting a subtree
assigned for node $i\in V$ and a potential function, respectively.
Formally, let us consider the following optimization problem DTB over subtree-valued potentials:
\begin{align}
\text{(DTB)\quad Maximize} \quad & \sum_{s\in S}r_s\dist(0,p_s)
-\sum_{i\in V\setminus S}c_i\size(p_i)
-\sum_{ij\in E}u_{ij}(\dist(p_i,p_j)-a_{ij})^+\notag\\
\text{subject to}\quad & p:V\rightarrow \mathbb{S},\notag\\
\label{cond:dtbfeas}
& p_s\in P_s\quad (s\in S).
\end{align}
It turns out, in the proof of \cref{prop:weakdual},
that this seemingly strange formulation of DTB is essentially the LP-dual of FNTB.
If $p:V\rightarrow \mathbb{S}$ satisfies \cref{cond:dtbfeas},
then it is called a \emph{potential}.
See \cref{fig:dtb} for an intuition of a subtree-valued potential $p$.
In the figure, $1\in S$ is mapped to a node on $P_1$,
and $i\in V\setminus S$ is mapped to a 1-type subtree.
A potential $p$ is said to be \emph{proper}
if the minimal subtree containing all $p_s\ (s\in S)$ also contains all $p_i\ (i\in V)$.
The following weak duality holds.

\begin{figure}[t]
\centering
\includegraphics[scale=0.6]{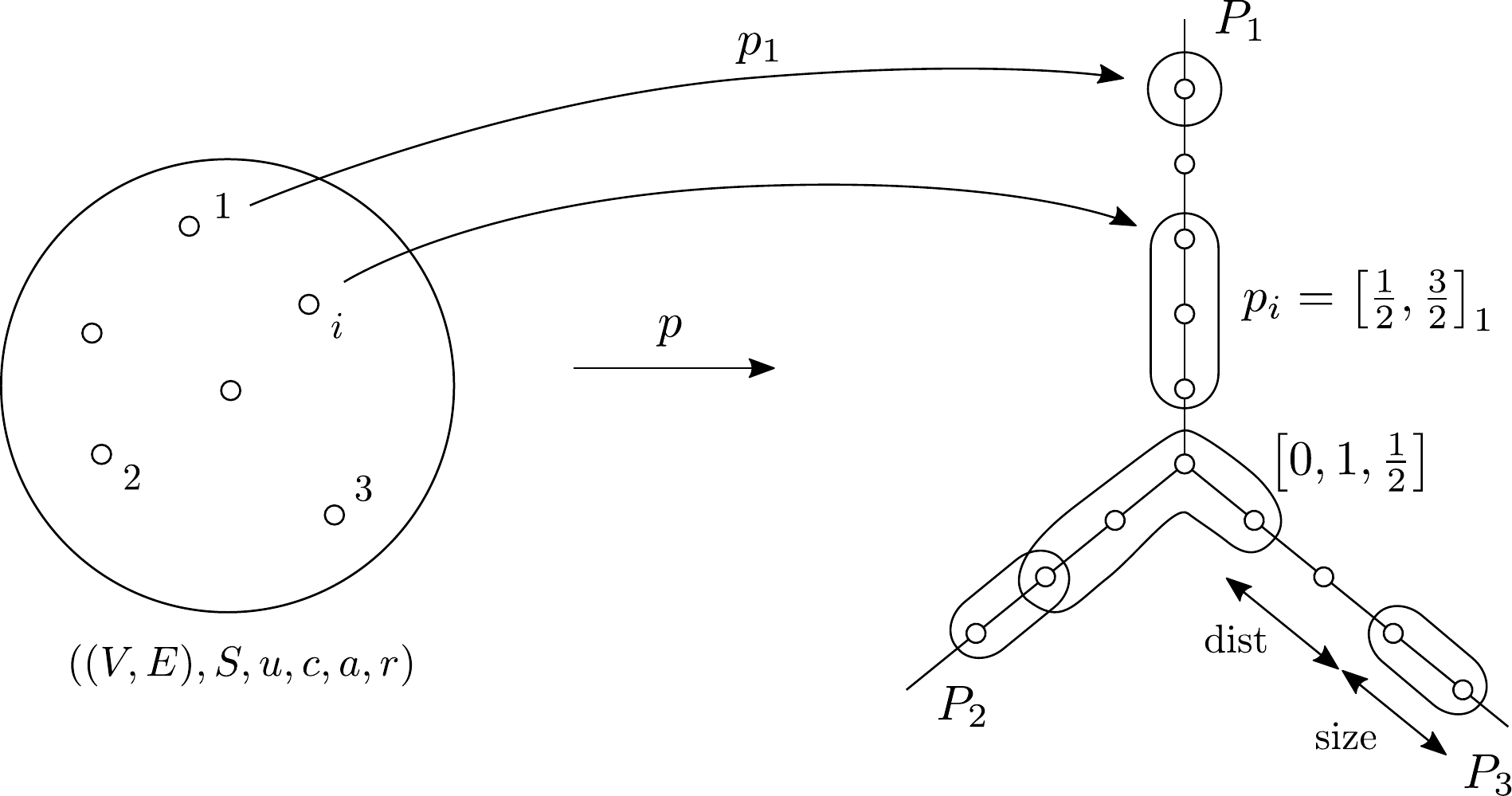}
\caption{A subtree-valued potential $p$.}
\label{fig:dtb}
\end{figure}

\begin{proposition}
\label{prop:weakdual}
The optimal value of FNTB is at least that of DTB.
Moreover, there exists a proper optimal potential for DTB.
\end{proposition}

\begin{proof}
Let $p:V\rightarrow\mathbb{S}$ be any potential (not necessarily proper).
For each $s\in S$, suppose that $p_s$ is written as $p_s=(L_s,s)$ where $L_s\in \mathbb{Z}^*_+$.
Define a new proper potential $p':V\rightarrow\mathbb{S}$ by
\[
    p'_i:=\begin{cases}
        [\min\{l,L_s\},\min\{l',L_s\}]_s & \text{if $p_i=[l,l']_s$ ($s$-type),}\\
        [\min\{l_1,L_1\},\dotsc,\min\{l_k,L_k\}] & \text{if $p_i=[l_1,\dotsc,l_k]$ (0-type).}
    \end{cases}
\]
Then the objective function does not decrease.
This implies the latter part of the statement.

We next show the former part, i.e., the weak duality.
The LP dual of FNTB is written as
\begin{align*}
\text{Maximize} \quad & \sum_{s\in S}\sum_{\hat{X}\in\mathcal{C}_s}
    (r_s-c(\Gamma(\hat{X})))\pi(\hat{X})
	-\sum_{e\in E}u_e\left(\sum_{\hat{X}\in\mathcal{C}:e\in\delta\hat{X}}\pi(\hat{X})-a_e\right)^+\\
\text{subject to}\quad & \pi:\mathcal{C}\rightarrow\mathbb{R}_+.
\end{align*}
We show that for any proper potential $p:V\rightarrow\mathbb{S}$,
we can construct $\pi:\mathcal{C}\rightarrow\mathbb{R}_+$ such that
\begin{align}
\label{eq:pot1}\sum_{\hat{X}\in \mathcal{C}_s} \pi(\hat{X})&=\dist(0,p_s)
\quad (s\in S),\\
\label{eq:pot2}\sum_{\hat{X}\in \mathcal{C}:i\in \Gamma(\hat{X})} \pi(\hat{X})&=\size(p_i)
\quad (i\in V\setminus S),\\
\label{eq:pot3}\sum_{\hat{X}\in \mathcal{C}:ij\in \delta \hat{X}} \pi(\hat{X})&=\dist(p_i,p_j)
\quad (ij\in E).
\end{align}
By $\sum_{\hat{X}\in\mathcal{C}}c(\Gamma(\hat{X}))\pi(\hat{X})
=\sum_{\hat{X}\in\mathcal{C}}\sum_{i\in \Gamma(\hat{X})}c_i\pi(\hat{X})
=\sum_{i\in V\setminus S}c_i\sum_{\hat{X}\in\mathcal{C}:i\in\Gamma(\hat{X})}\pi(\hat{X})$,
the weak duality follows from \eqref{eq:pot1}--\eqref{eq:pot3}.

Let $e$ be an edge in $\mathbb{T}$.
We define a biset $(X_e,X_e^+)$ as follows.
When we remove $e$ from $\mathbb{T}$, there appear two connected components.
Let $T_e$ be the component which does not contain $0\in \mathbb{T}$.
Let $X_e$ be the set of nodes $i\in V$ such that $p_i$ is contained in $T_e$.
Let $X'_e$ be the set of nodes $i\in V$ such that
 $p_i$ is not contained in the component $\mathbb{T}\setminus (T_e\cup\{e\})$.
Since $p$ is proper, if $e$ is an edge in $P_s$ and $X_e\neq\emptyset$, then $(X_e,X_e^+)\in \mathcal{C}_s$.
Observe that $p_i$ contains edge $e$ if and only if $i\in \Gamma(X_e,X^+_e)$.
Also $e$ belongs to the shortest path between $p_i$ and $p_j$
if and only if $i\in X_e$ and $j\notin X^+_e$, or $i\notin X^+_e$ and $j\in X_e$,
i.e., $ij\in \delta(X_e,X^+_e)$.
Thus a potential function $\pi:\mathcal{C}\rightarrow\mathbb{R}_+$ defined by
\[
    \pi(\hat{X}):=\frac{1}{2}\lvert \{e\mid \hat{X}=(X_e,X_e^+)\}\rvert\quad (\hat{X}\in\mathcal{C})
\]
satisfies \eqref{eq:pot1}--\eqref{eq:pot3}.
\end{proof}

We remark that the technique used in the above proof is based on a tree representation of a laminar biset family;
see also \cite{Hirai2013Half} for related arguments.
Our algorithm below will give an algorithmic proof of the strong duality.

We next derive from \cref{prop:weakdual} the complementary slackness condition.
Let $p:V\to \mathbb{S}$ be a proper potential.
According to $p$, we decompose $V$ into $S\cup V_0\cup \bigcup_{s\in S}V_s$, where
\begin{align*}
    V_0&:=\{i\in V\setminus S \mid \text{$p_i$ is of 0-type}\},\\
    V_s&:=\{i\in V\setminus S \mid \text{$p_i$ is of $s$-type}\}.
\end{align*}
In the next lemma, we see that we may only consider edges $ij\in E$
satisfying $\dist(p_i,p_j)\geq a_{ij}$ (by the condition (C2)).
Let denote the set of such edges by
\[
    E_p:=\{ij\in E\mid \dist(p_i,p_j)\geq a_{ij}\}.
\]
For $i\in V_0$ and $s\in S$,
let $\delta_s i$ be the set of edges $ij$ in $\delta i \cap E_p$
with $j \in V_s$.
For $i\in V_s\ (s\in S)$,
there appear two connected components
when we remove $p_i$ from $\mathbb{T}$;
one includes $0\in \mathbb{T}$, and the other does not.
By \textbf{(CP)},
for a node $j\in V$ with $ij\in \delta i\cap E_p$,
its potential $p_j$ is contained in one of the two components.
Let $\delta_0 i$ be the set of edges $ij$ in $\delta i \cap E_p$
such that $p_j$ is in the component including $0$.
Let $\delta_s i:=(\delta i\cap E_p)\setminus \delta_0 i$.
See \cref{fig:dtb_pot} for the definition.

\begin{figure}[t]
\centering
\includegraphics[scale=0.62]{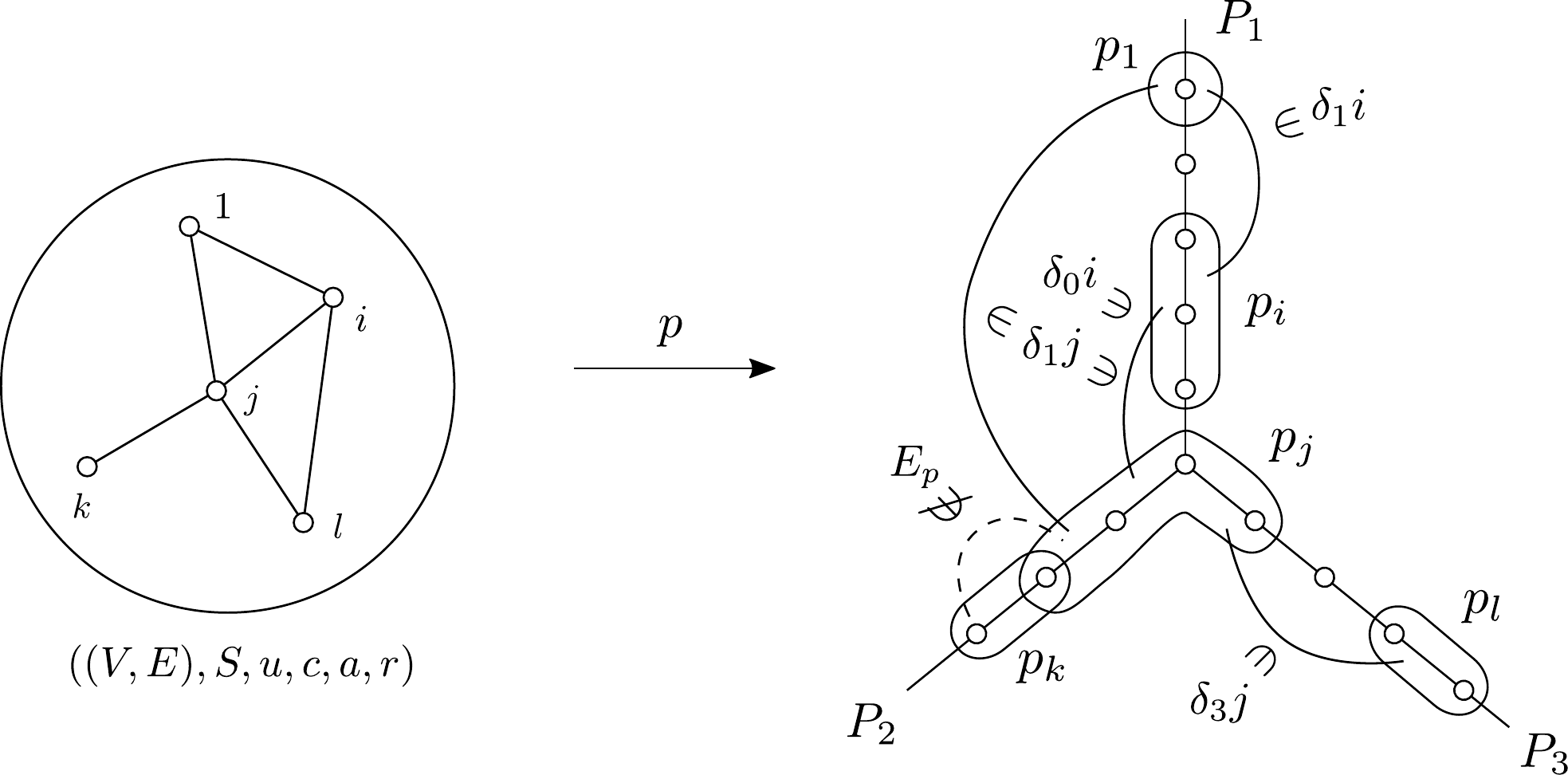}
\caption{$\delta_0 i$ and $\delta_1 i$ for $i\in V_1$
and $\delta_s j\ (s\in S)$ for $j\in V_0$.}
\label{fig:dtb_pot}
\end{figure}

\begin{lemma}
\label{lem:slack}
Let $x:E\rightarrow\mathbb{R}_+$ be
an edge-capacity function with $0\leq x\leq u$,
and let $p:V\rightarrow\mathbb{S}$ be a proper potential.
If $x$ and $p$ satisfy the following conditions (C1--5),
then $x$ and $p$ are optimal solutions for FNTB and DTB, respectively:

\begin{enumerate}
\renewcommand{\labelenumi}{\textup{(C\arabic{enumi})}}
\item For each $ij\in E$, if $\dist(p_i,p_j)>a_{ij}$, then $x_{ij}=u_{ij}$.
\item For each $ij\in E$, if $\dist(p_i,p_j)<a_{ij}$, then $x_{ij}=0$.
\item For each $i\in V_s\ (s\in S)$, it holds $x(\delta_0 i)=x(\delta_s i)\leq c_i$.
    If $\size(p_i)>0$, then $x(\delta_0 i)=x(\delta_s i)=c_i$.
\item For each $i\in V_0$ and $s\in S$, it holds $x(\delta_s i)\leq c_i$
	 and $x(\delta_s i)\leq \sum_{s'\in S-s} x(\delta_{s'} i)$.
    If $\size_s(p_i)>0$, then $x(\delta_s i)=c_i$.
\item For each $s\in S$, it holds $x(\delta s)\geq r_s$.
    If $\dist(0,p_s)>0$, then $x(\delta s)=r_s$.
\end{enumerate}
\end{lemma}

\begin{proof}
Suppose $x$ and $p$ satisfy (C1--5).
For the feasibility of $x$, it is sufficient to show that,
for each $s\in S$,
there exists a flow satisfying the capacities $x$ and $c$
that connects $s$ and $S-s$ with flow-value $r_s$.
We prove a stronger result that $x$ can be decomposed into
 a separately-capacitated multiflow with sufficient flow-values.
An \emph{$S$-path} is a path connecting distinct terminals.
See the algorithm \textsc{Decompose} (\cref{algo:decompose}),
 which takes $x$ as an input
and outputs a function $\lambda:\mathcal{P}\rightarrow\mathbb{R}_+$,
where $\mathcal{P}$ is a set of $S$-paths.

\begin{algorithm}[t]
\caption{\textsc{Decompose}}
\label{algo:decompose}
\begin{enumerate}
\setcounter{enumi}{-1}
\item Let $\mathcal{P}=\emptyset$.
\item Take $s_0\in S$ and an edge $s_0 j$ satisfying $x(s_0 j)>0$.
 If such a pair does not exist, then stop;
 output $(\mathcal{P},\lambda)$.
 Otherwise, let $j_0\leftarrow s_0,\ j_1\leftarrow j,\ \mu\leftarrow x(s_0 j),\ 
 t\leftarrow 1$.
\item If $j_t$ is a terminal, then add $P=(j_0,j_1,\dotsc,j_t)$ to $\mathcal{P}$
 and let $\lambda(P):=\mu>0$.
 Update $x(e)\leftarrow x(e)-\mu$ on each edge $e$ in $P$, and return to Step 1.
 Otherwise go to Step 3.
\item Suppose $j_t\in V_s\ (s\in S)$.
 Since $x(j_{t-1}j_t)>0$, it holds $j_{t-1}j_t\in \delta_s {j_t}$
 or $j_{t-1}j_t\in \delta_0 {j_t}$ by (C2).
 If the former is the case, take $j_tj_{t+1}\in \delta_0 {j_t}$ with $x(j_tj_{t+1})>0$.
 Such an edge exists by the former part of (C3).
 If the latter is the case, take $j_tj_{t+1}\in \delta_s {j_t}$
 with $x(j_tj_{t+1})>0$.
 Update $\mu\leftarrow \min\{\mu,x(j_tj_{t+1})\},\ t\leftarrow t+1$,
 and return to Step 2.

 Suppose $j_t\in V_0$. As we will show later, it holds $j_{t-1}j_t\in \delta_{s_0} {j_t}$.
 Take $s\in S-s_0$ with maximum $x(\delta_{s} {j_t})\ (>0)$,
 and take $j_tj_{t+1}\in \delta_{s} {j_t}$ with $x(j_tj_{t+1})>0$.
 Such an edge exists by $x(j_{t-1}j_t)>0$ and the former part of (C4).
 Update
 \[
    \mu\leftarrow\min\left\{\mu,x(j_tj_{t+1}),
            \frac{1}{2}\min_{s'\in S\setminus \{s_0,s\}}
			\left\{\sum_{s''\in S-s'} x(\delta_{s''} {j_t})-x(\delta_{s'} {j_t})\right\}\right\},
\]
 and $t\leftarrow t+1$.
 Note that $\mu>0$ by the maximality of $x(\delta_{s} {j_t})$.
 Return to Step 2.
\end{enumerate}
\end{algorithm}

\begin{claim}
The number of iterations of \textsc{Decompose} is at most $O(m+kn)$.
Each iteration can be done in $O(n)$ time.
\end{claim}

\begin{proof}
Suppose that we add $(j_0,j_1,\dotsc,j_\ell)$ to $\mathcal{P}$ in Step 2.
For each $t=1,\dotsc,\ell-1$, observe that
$p_{j_{t+1}}$ is at a side opposite to $p_{j_{t-1}}$ based on $p_{j_{t}}$.
By \textbf{(CP)} and (C2),
$\{j_{t-1},j_{t},j_{t+1}\}$ are distinct.
Thus we have
\[
\dist(p_{j_{t-1}},p_{j_{t+1}})=
\begin{dcases*}
	\dist(p_{j_{t-1}},p_{j_t})+\size(p_{j_{t}})+\dist(p_{j_{t}},p_{j_{t+1}})
		& if $j_{t}\in \bigcup_{s\in S}V_s$,\\
	\dist(p_{j_{t-1}},p_{j_t})+\size_{s}(p_{j_{t}})+\size_{s'}(p_{j_{t}})+\dist(p_{j_{t}},p_{j_{t+1}})
		& if $j_{t}\in V_0$,
\end{dcases*}
\]
where $j_{t-1}\in V_s$ and $j_{t+1}\in V_{s'}$ ($s\neq s'$)
for the case $j_t\in V_0$.
Since $\mathbb{T}$ is a tree and $\dist(p_{j_{t}},p_{j_{t+1}})>0$ for each $t$,
we can see
\begin{equation}
\label{eq:shortestpath}
    \dist(p_{j_0},p_{j_\ell})=\sum_{0\leq t\leq \ell-1} \dist(p_{j_{t}},p_{j_{t+1}})+\sum_{1\leq t\leq \ell-1,\ t\neq t'}\size(p_{j_{t}})
        +\size_{j_0}(p_{j_{t'}})+\size_{j_\ell}(p_{j_{t'}}),
\end{equation}
where $j_{t'}\in V_0$ (if such $t'$ exists); see also \cite[Lemma 3.9]{Hirai2015L}.
Hence $(j_0,j_1,\dotsc,j_\ell)$ is a ``shortest path on $\mathbb{T}$'' from $j_0$ to $j_\ell$,
and $j_0,\dotsc,j_\ell$ are all distinct.

By the distinctness, Step 3 is executed at most $\lvert V\rvert$ times
per an $S$-path $P$.
Also the algorithm keeps (C2) and the former parts of (C3--4).
To see it for (C4), suppose that the algorithm adds an $S$-path $(j_0,j_1,\dotsc,j_t,\dotsc,j_\ell)$ to $\mathcal{P}$ in Step 2,
where $j_0=s\in S$, $j_t\in V_0$ and $j_\ell=s'\in S$.
By the above argument, such $t$ is uniquely determined (if exists).
Then for all $s''\neq s,s'$, we have $\sum_{s'''\in S-s''}x(\delta_{s'''} {j_t})-x(\delta_{s''} {j_t})\geq 2\mu$.
Thus after the decrease of the value of $x$ along with $P$, it satisfies
$\sum_{s'''\in S-s''}x(\delta_{s'''} {j_t})-x(\delta_{s''} {j_t})\geq 0$.

After the decrease along with an $S$-path,
it becomes $x(e)=0$ for at least one edge $e\in E$,
or becomes $\sum_{s'\in S-s}x(\delta_{s'} {i})-x(\delta_s {i})=0$
for at least one pair of $i\in V_0$ and $s\in S$.
The algorithm keeps those values to be zero in the remaining execution,
implying that it terminates after adding at most $O(m+kn)$ paths to $\mathcal{P}$.
To see it, suppose that after the decrease,
it becomes $\sum_{s'\in S-s}x(\delta_{s'} {i})-x(\delta_s {i})=0$ for $i\in V_0$ and $s\in S$.
If the algorithm chooses an $S$-path $(j_0,\dotsc,j_t=i,\dotsc,j_\ell)$ for adding to $\mathcal{P}$
in the remaining execution,
by the maximality of $x(\delta_s i)$,
 it should satisfy that $j_{t-1}j_t\in \delta_s {i}$
or $j_tj_{t+1}\in \delta_s {i}$.
Thus $\sum_{s'\in S-s}x(\delta_{s'} {i})-x(\delta_s {i})$ does not change
 by the decrease along with $(j_0,\dotsc,j_\ell)$.
\end{proof}

We next show that the output $f=(\mathcal{P},\lambda)$
is a separately-capacitated multiflow of the network $((V,E),S,x,c)$.
We first observe that any edge $e\in E$ satisfies $x(e)=0$ at the end of \textsc{Decompose}.
In fact, if there exists an edge $e\in E$ with $x(e)> 0$,
then we can construct an $S$-path with edges having positive $x$-values
by repeating to apply the former parts of (C3--4).

Let $f(e):=\sum_{P\in \mathcal{P}:e\in P}\lambda(P)$ for $e\in E$,
and let $f(i):=\sum_{P\in\mathcal{P}:i\in P}\lambda(P)$ for $i\in V$.
Also let $\mathcal{P}_s\subseteq \mathcal{P}$ be
the subset of paths connecting $s$ to other terminals,
and let $f_s=(\mathcal{P}_s,\lambda_s:=\lambda|_{\mathcal{P}_s})$ for $s\in S$.
Clearly, $f$ satisfies the edge-capacity (by $f(e)=x(e)$).
For $i\in V_s\ (s\in S)$, if a path $P\in\mathcal{P}$ goes through $i$,
then $P$ must be contained in $\mathcal{P}_s$.
Thus by the former part of (C3),
 $f_s(i)=f(i)\leq x(\delta_0 {i})\ (=x(\delta_s {i})) \leq c(i)$,
and $f_{s'}(i)\leq f_s(i)\leq c(i)$ for any other $s'\in S-s$.
For $i\in V_0$,
 if a path in $\mathcal{P}_s\ (s\in S)$ goes through $i$,
then it must do that via an edge in $\delta_s {i}$.
Thus by the former part of (C4), we have $f_s(i)\leq x(\delta_s {i})\leq c(i)$.
Hence $f$ is a separately-capacitated multiflow.
Moreover, $f_s$ satisfies the requirement $r$ by the former part of (C5)
(and $f(e)=x(e)$ for all $e\in E$).
This proves that $x$ is a feasible solution of FNTB.

We next show the optimality of $x$ and $p$.
We see that
\begin{align}
&\quad\sum_{ij\in E} a_{ij}x_{ij}
    -\sum_{s\in S}r_s\dist(0,p_s)
    +\sum_{i\in V\setminus S}c_i\size(p_i)
	+\sum_{ij\in E}u_{ij}(\dist(p_i,p_j)-a_{ij})^+\notag\\
&=\sum_{ij\in E} (\dist(p_i,p_j)-a_{ij})^+(u_{ij}-x_{ij})
	+\sum_{ij\in E} (a_{ij}-\dist(p_i,p_j))^+ x_{ij}
	+\sum_{ij\in E} x_{ij}\dist(p_i,p_j)\notag\\
&\quad +\sum_{i\in V\setminus S}c_i\size(p_i)
    -\sum_{s\in S}r_s\dist(0,p_s)\notag\\
&=\sum_{ij\in E} (\dist(p_i,p_j)-a_{ij})^+(u_{ij}-x_{ij})
	+\sum_{ij\in E} (a_{ij}-\dist(p_i,p_j))^+x_{ij}\notag\\
&\quad +\sum_{s\in S}\sum_{i\in V_s}(c_i-f(i))\size(p_i)
	+\sum_{i\in V_0}\sum_{s\in S}(c_i-f_s(i))\size_s(p_i)
    +\sum_{s\in S}(f(s)-r_s)\dist(0,p_s).\label{eq:pridual_diff}
\end{align}
The first equality comes from $a+(d-a)^+=d+(a-d)^+$ for $a,d\in \mathbb{R}$.
The second equality comes from \cref{eq:shortestpath} and
\begin{align*}
&\quad\sum_{ij\in E} f(ij)\dist(p_i,p_j)
    +\sum_{s\in S}\sum_{i\in V_s}f(i)\size(p_i)
    +\sum_{i\in V_0}\sum_{s\in S}f_s(i)\size_s(p_i)\\
&=\sum_{ij\in E}\sum_{P\in\mathcal{P},ij\in E(P)} \lambda(P)\dist(p_i,p_j)\\
    &\qquad\qquad+\sum_{s\in S}\sum_{i\in V_s}\sum_{P\in\mathcal{P},i\in V(P)} \lambda(P)\size(p_i)
    +\sum_{i\in V_0}\sum_{s\in S}\sum_{P\in\mathcal{P}_s,i\in V(P)} \lambda_s(P)\size_s(p_i)\\
&=\sum_{st}\sum_{P\in\mathcal{P}_s\cap\mathcal{P}_t}
    \lambda(P)\dist(p_s,p_t)
    =\sum_{s\in S}f(s)\dist(0,p_s).
\end{align*}
We see $f(i)=x(\delta_0 {i})$ ($=x(\delta_s {i})$)
for $i\in V_s\ (s\in S)$,
and $f_s(i)=x(\delta_s {i})$ for $i\in V_0$ and $s\in S$.
Also $f(s)=x(\delta s)$ for $s\in S$.
Then \cref{eq:pridual_diff} is zero by (C1--2) and the latter parts of (C3--5).
By \cref{prop:weakdual}, we conclude that $x$ and $p$ are both optimal.
\end{proof}

The algorithm \textsc{Decompose} is based on \cite[Lemma 4.5]{Hirai2013Half};
see also \cite[Lemma 3.3]{Hirai2018Dual}.

\begin{corollary}
\label{cor:halfflow}
Let $x$ be a solution of FNTB satisfying (C1--5)
 with a potential for DTB.
Then there exists a separately-capacitated multiflow $f$
 such that $f(e)=x(e)$ for all $e\in E$ and $f(s)\geq r_s$ for all $s\in S$.
If $x$ is half-integral and $x(\delta i)\in \mathbb{Z}_+$ for any $i\in V$,
then $f$ can be taken as half-integer-valued.
Given $x$, such a flow $f$ can be obtained in $O((m+kn)n)$ time.
\end{corollary}

\begin{proof}
For such $x$,
 \textsc{Decompose} always takes $\mu$ as half-integral,
and keeps the half-integrality of $x$ and the integrality of $x(\delta i)$ for all $i\in V$.
Thus the output is a half-integer-valued multiflow.
\end{proof}

As we will see in the next subsection,
 the existence of an edge-capacity $x$ satisfying (C1--5)
can be checked by solving an undirected circulation problem.
This suggests a simple descent algorithm \textsc{Descent} for DTB and FNTB.
Notice that a potential $p:V\to \mathbb{S}$ can be identified with a vector in $\mathbb{S}^n$.
For brevity, we write $p\in \mathbb{S}^n$ below.
Let $h:\mathbb{S}^n\rightarrow \overline{\mathbb{R}}$ be
a function defined by
\begin{equation}
\label{eq:horig}
    h(p):=\begin{dcases*}
		-\sum_{s\in S}r_s\dist(0,p_s)+\sum_{i\in V\setminus S}c_i\size(p_i)\\
		\qquad +\sum_{ij\in E}u_{ij}(\dist(p_i,p_j)-a_{ij})^+
		 & if $p\in \mathbb{S}^n$ is a potential,\\
		\infty & otherwise.
	\end{dcases*}
\end{equation}
Then DTB is precisely minimization of $h$ over $\mathbb{S}^n$.

\textsc{Descent} starts with the initial potential $p\equiv 0$
(i.e., $p_i=0\in\mathbb{S}$ for any $i\in V$).
At the beginning of each iteration,
check the sufficiency of the optimality of $p$
 by finding $x$ satisfying (C1--5).
If $x$ is found, then $x$ and $p$ are optimal.
Otherwise find $q\in \mathbb{S}^n$ with $h(q)<h(p)$,
update $p$ by $q$, and repeat this iteration.
It turns out that $q$ is also computed by the (same) undirected circulation problem.
In the following subsections, we introduce the undirected circulation problem
and discuss how to find $x$ or $q$ in each iteration.
We will give the detail of \textsc{Descent} in \cref{subsec:algo}.

\subsection{Checking Optimality}
\label{subsec:opt}

For a given proper potential $p\in\mathbb{S}^n$,
the existence of $x:E\to \mathbb{R}_+$ satisfying (C1--5)
reduces to an undirected circulation problem on the following network 
$\mathcal{N}_p:=((U,F),\underline{c},\overline{c})$.
See \cref{fig:bisubflow} for the construction.

\begin{itemize}
\item For each $i\in V_s\ (s\in S)$,
 divide $i$ into a two-node set $U_i:=\{i^0,i^s\}$,
 and connect the nodes by an edge $e_i:=i^0i^s$.
For representing (C3), let
\[
	\underline{c}(e_i):=-c_i,\quad
	\overline{c}(e_i):=\begin{dcases*}
		0 & if $\size(p_i)=0$,\\
		-c_i & if $\size(p_i)>0$.
	\end{dcases*}
\]
\item For each $i\in V_0$, divide $i$ into a $2k$-node set $\tilde{U_i}:=U_i^0\cup U_i$,
where $U_i^0:=\{i^{1,0},i^{2,0},\dotsc,i^{k,0}\}$ and
$U_i:=\{i^{1},i^{2},\dotsc,i^{k}\}$.
For all $s\in S$, connect $i^{s,0}$ and $i^{s}$
by an edge $e_i^s:=i^{s,0}i^{s}$ (for representing (C4)) with capacity
\[
	\underline{c}(e_i^s):=-c_i,\quad
	\overline{c}(e_i^s):=\begin{dcases*}
		0 & if $\size_s(p_i)=0$,\\
		-c_i & if $\size_s(p_i)>0$
	\end{dcases*}\quad (s\in S).
\]
Connect each pair of nodes in $U_i^0$ by an edge of
lower capacity 0 and upper capacity $\infty$.

\item For each $s\in S$, let $s^0:=s$ and $U_s:=\{s^0\}$,
and add a self-loop $e_s:=s^0s^0$.
For representing (C5), let
\[
	\underline{c}(e_s):=\begin{dcases*}
		-\infty & if $\dist(0,p_s)=0$,\\
		-r_s & if $\dist(0,p_s)>0$,
	\end{dcases*}
	\quad\overline{c}(e_s):=-r_s.
\]
\end{itemize}

\begin{figure}[t]
\centering
\includegraphics[scale=0.9]{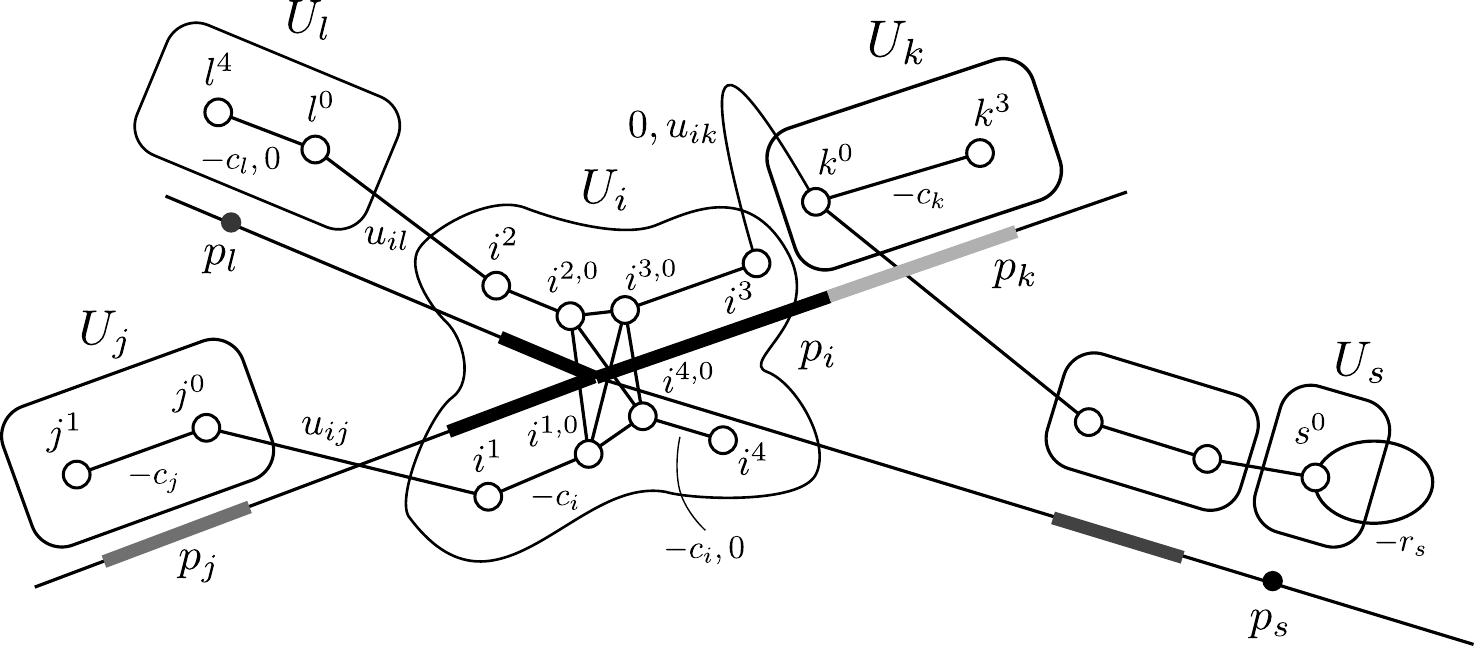}
\caption{The undirected network $\mathcal{N}_p$.}
\label{fig:bisubflow}
\end{figure}

We replace endpoints of each edge $ij\in E_p$ as follows.

\begin{itemize}
\item If $i\in V_0$ and $j\in V_s$, then replace $ij$ with $i^{s}j^0$.
\item If $i\in V_s$ and $j\in V_{s'}$ ($s\neq s'$),
then replace $ij$ with $i^0j^0$.
\item If $i,j\in V_s$ and $p_i$ is closer to 0 than $p_j$,
i.e., $\dist(0,p_i)<\dist(0,p_j)$, then replace $ij$ with $i^sj^0$.
\end{itemize}
For each replaced edge $e$ coming from $ij\in E_p$, let
\[
	\underline{c}(e):=\begin{dcases*}
		0 & if $\dist(p_i,p_j)=a_{ij}$,\\
		u_{ij} & if $\dist(p_i,p_j)>a_{ij}$,
	\end{dcases*}
	\quad\overline{c}(e):=u_{ij}.
\]

The node set $U$ and the edge set $F$ are defined as the union of all nodes and edges
in the above construction, respectively.

\begin{theorem}
\label{thm:bisub_opt}
Let $\mathcal{N}_p=((U,F),\underline{c},\overline{c})$ be
the undirected network
constructed from a proper potential $p:V\rightarrow\mathbb{S}$.
Suppose that $\mathcal{N}_p$ has a (half-integer-valued)
circulation $y:F\rightarrow\mathbb{R}$.
Then an edge-capacity function $x:E\rightarrow\mathbb{R}_+$ defined by
\[
    x(e):=\begin{cases}
        y(e) & \text{if $e\in E_p$},\\
        0 & \text{otherwise}
    \end{cases}
\]
satisfies (C1--5).
\end{theorem}

\begin{proof}
We can check (C1--5) from the construction of $\mathcal{N}_p$ easily.
For example, the former part of (C4) follows from
$x(\delta_s i)=-y(i^{s,0}i^{s})\leq -\underline{c}(i^{s,0}i^{s})=c_i$ and
\[
    x(\delta_s i)=-y(i^{s,0}i^{s})
    =\sum_{s'\neq s} y(i^{s,0}i^{s',0})
    \leq \sum_{s'\neq s} -y(i^{s',0}i^{s'})
    =\sum_{s'\neq s} x(\delta_{s'} i).
\]
Also the latter part of (C4) follows from
$x(\delta_s i)=-y(i^{s,0}i^{s})\geq -\overline{c}(i^{s,0}i^{s})=c_i$
for $i\in V_0$ and $s\in S$ with $\size_s(p_i)>0$.
\end{proof}

\subsection{Finding a Descent Direction}
\label{subsec:direction}

If the algorithm in \cref{lem:feascirc} outputs a circulation
in $\mathcal{N}_p$, then an optimal edge-capacity is computed from the circulation,
and $p$ is optimal by \cref{lem:slack,thm:bisub_opt}.
Otherwise the algorithm outputs a maximum violating cut.
We can find $q\in\mathbb{S}^n$ with $h(q)<h(p)$
using the maximum violating cut.
This implies the necessity of \cref{lem:slack}
and the strong duality of \cref{prop:weakdual}.

We find such a  $q\in\mathbb{S}^n$ with the property that
 that $p_i$ is of $s$-type (i.e., $p_i\subseteq P_s\setminus\{0\}$)
implies $q_i\subseteq P_s$.
Let $\tilde{U} :=\bigcup_{i\in V} U_i$.
In the view of \cref{subsec:space_of}, such a subtree-valued potential 
$q \in \mathbb{S}^n$ is identified with $q \in  (\mathbb{Z}^*)^{\tilde U}$ by the following correspondence:
\begin{align}
q_s=(l,s) & \Leftrightarrow q(s^0)=-l \leq 0 \quad (s \in S), \label{eq:s_in_S}\\
  q_i = [l,l']_s & \Leftrightarrow q(i^s) =l'\geq l=-q(i^0) \geq 0 \quad (i \in V_s,\ s\in S), \label{eq:i_in_Vs} \\
 q_i =(l_s)_{s \in S} & \Leftrightarrow q(i^s) = l_s \quad (i \in V_0,\ s \in S), \label{eq:i_in_V0}
\end{align}
where 
$(l_s)_{s \in S} = (q(i^s))_{s \in S}$ must 
satisfy \cref{eq:subtree} in \cref{eq:i_in_V0}. 
The correspondence \eqref{eq:s_in_S}--\eqref{eq:i_in_V0}
can be considered as 
a projection from $(\mathbb{Z}^*)^{kn}$ to $(\mathbb{Z}^*)^{\tilde{U}}$.
We will particularly care with the ``becoming $s$-type'' move 
$(0,\ldots,0, l_s,0,\ldots,0) \to (- 1/2,\ldots,-1/2,l_s',-1/2,\ldots,-1/2)$.

The movement from $p$ to $q$ is done
according to the intersection pattern of
a maximum violating cut $(Y,Z)$ with $U_i$.
For a cut $(Y,Z)\in 3^U$, let $\tilde{\chi}_{Y,Z}$ denote
 the restriction of $\chi_{Y,Z}\in \mathbb{Z}^U$ to $\tilde{U}$.
Then we can consider a vector
\begin{equation}
p^{Y,Z}:=p+ \frac{1}{2} \tilde{\chi}_{Y,Z} \in (\mathbb{Z}^*)^{\tilde{U}}.
\end{equation}
A cut  $(Y,Z)$ is said to be \emph{movable}
if $p^{Y,Z} \in ({\mathbb{Z}^*})^{\tilde U}$ is regarded as a potential in $\mathbb{S}^n$
by the correspondence \eqref{eq:s_in_S}--\eqref{eq:i_in_V0}.
Since no condition on $U_i^0$
is imposed, different movable cuts
may determine the same subtree-valued potential.
To eliminate such redundancy, 
we additionally impose the following normalization condition on movable cuts $(Y, Z)$:
\begin{itemize}
\item
For each $i \in V_0$, if $(p^{Y,Z})_i$ becomes $s$-type, 
then $i^{s,0} \in Z$ and $i^{s',0} \in Y$ for all $s' \in S-s$.
Otherwise $Y\cap U_i^0=Z\cap U_i^0=\emptyset$.
\end{itemize}
We will see in the proof of Theorem~\ref{thm:map}
that we can make $(Y,Z)$ satisfy this condition without decreasing 
the cut value $\kappa$. 

We further introduce two types of movable cuts.
Let define a partition $(U^\uparrow,U^\downarrow)$ of $U$ by
\begin{align*}
    U^\uparrow&:=\{j\in\tilde{U} \mid \text{$p(j)$ is integral}\}
		        \cup \textstyle\bigcup_{i\in V_0} U_i^0,\\
	U^\downarrow&:=\{j\in\tilde{U} \mid \text{$p(j)$ is non-integral}\},
\end{align*}
where $p$ is regarded as $p:\tilde U \to \mathbb{Z}^*$
 by \eqref{eq:s_in_S}--\eqref{eq:i_in_V0}.
A movable cut $(Y,Z)\in 3^U$ is \emph{upward-movable} (resp. \emph{downward-movable})
if $Y\cap U^\downarrow=Z\cap U^\downarrow=\emptyset$
(resp. $Y\cap U^\uparrow=Z\cap U^\uparrow=\emptyset$).
Let  $\mathcal{M}^\uparrow$ and $\mathcal{M}^\downarrow$ denote the sets of all upward-movable cuts and downward-movable cuts, respectively.
An upward-movable (resp. downward-movable) cut $(Y,Z)\in 3^U$
is \emph{maximum} if it attains the maximum $\kappa(Y,Z)$ among
all upward-movable (resp. downward-movable) cuts.
The goal of this subsection is to show the following.

\begin{theorem}
\label{thm:map}
\leavevmode
\begin{enumerate}
\renewcommand{\labelenumi}{\textup{(\arabic{enumi})}}
\item For $(Y,Z)\in \mathcal{M}^\uparrow\cup \mathcal{M}^\downarrow$,
    it holds $h(p^{Y,Z})-h(p)=-\kappa(Y,Z)/2$.
\item Given a maximum violating cut, 
we can compute in $O(kn)$ time, a maximum upward-movable cut and 
a maximum downward-movable cut, 
where at least one of these movable cuts is violating.
\end{enumerate}
\end{theorem}

\begin{proof}
(1) Let $(Y,Z)\in\mathcal{M}^\uparrow\cup\mathcal{M}^\downarrow$, and let $p':=p^{Y,Z}$.
Let
\begin{equation}
\label{eq:cutvaledge}
\kappa_{uv}(Y,Z):=\chi_{Y,Z}(u,v)^+ \underline{c}(uv)
        -\chi_{Z,Y}(u,v)^+\overline{c}(uv).
\end{equation}
Then $\kappa(Y,Z) = \sum_{uv \in F} \kappa_{uv}(Y,Z)$. 
We will estimate each term $\kappa_{uv}(Y,Z)$.
We first observe from the vector representation of subtrees
that for nodes $i,j\in V$ it holds
\[
\dist (p_i, p_j) > 0 \Rightarrow  \dist (p_i, p_j) = - p(i^{\sigma}) - p(j^{\sigma'})
\]
for appropriate $\sigma,\sigma' \in S \cup \{0\}$. 
Since $(Y,Z)\in\mathcal{M}^\uparrow\cup\mathcal{M}^\downarrow$,
if $\dist (p_i, p_j) > 0$ is non-integral,  then 
one of $p(i^{\sigma})$ and $p(j^{\sigma})$ does not change in the move from $p$ to $p'$.
Consequently, 
the sign of  $\dist (p_i, p_j) - a_{ij}$ does not invert. 
Thus, for an edge $ij \in E_{p}$, we have 
\begin{align}
\label{eq:ij_cost}
&\quad u_{ij}(\dist(p'_i,p'_j) - a_{ij})^+ - u_{ij} (\dist(p_i,p_j) - a_{ij})^+ \notag \\
&=\begin{dcases*}
u_{ij} \chi_{Z,Y}(i^{\sigma},j^{\sigma'})^+/2 &  if $\dist(p_i,p_j)=a_{ij}$,\\
- u_{ij} \chi_{Y,Z}(i^{\sigma},j^{\sigma'})/2 &  if $\dist(p_i,p_j)>a_{ij}$
\end{dcases*}\notag\\
&= - \kappa_{i^{\sigma}j^{\sigma'}}(Y,Z)/2.
\end{align}
For $i \in V_s\ (s\in S)$, since $\size (p_i) = p(i^0) + p(i^s)$ it holds
\begin{equation}\label{eq:i_in_Vs_cost}
c_i \size (p_i') - c_i \size (p_i) = c_i \chi_{Y,Z}(i^0,i^s)/2 = - \kappa_{i^0i^s}(Y,Z)/2.
\end{equation}
The second equality can be seen as follows:
If $\size(p_i) > 0$, then $\overline{c}(e_i) = \underline{c}(e_i)= - c_i$.
If $\size(p_i) = 0$, i.e.,  $p(i^{s}) = - p(i^{0})$, 
then  $\underline{c}(e_i) = - c_i$, $\overline{c}(e_i)= 0$,
and $\chi_{Y,Z}(i^{0},i^{s}) \geq 0$ by the movability \eqref{eq:i_in_Vs}.
For both cases, it holds
 $-\underline{c}(e_i)\chi_{Y,Z}(i^{0},i^{s})^+
 + \overline{c}(e_i) \chi_{Z,Y}(i^{0},i^{s})^+ = c_i \chi_{Y,Z}(i^0,i^s)$.
Consider $i \in V_0$. Suppose that $p'_i$ is of 0-type, i.e., $p'(i^s) \geq 0$ $(s \in S)$.
Then $\size (p'_i) = \sum_{i \in S} p'(i^s)$, and hence 
\begin{equation}\label{eq:i_in_V0_cost_1}
c_i \size (p_i') - c_i \size (p_i)  = c_i \chi_{Y,Z}(U_i)/2 = 
- \sum_{s \in S} \kappa_{i^s i^{s,0}}(Y,Z)/2.
\end{equation}
The second equality can be seen as follows:
By the normalization condition, it holds $\chi_{Y,Z}(i^{s}) = \chi_{Y,Z}(i^{s,0}, i^{s})$.
Since $\size_{s} p_i = 0$ implies $\chi_{Y,Z}(i^{s,0},i^{s}) \geq 0$, as above,  we have~\eqref{eq:i_in_V0_cost_1}.
In the case where $p'_i$ becomes $s$-type, it holds 
\begin{equation}\label{eq:i_in_V0_cost_2}
c_i \size (p_i') - c_i \size (p_i)  =  c_i (\chi_{Y,Z}(i^s) - 1)/2 = - \sum_{s \in S} \kappa_{i^s i^{s,0}}(Y,Z)/2.
\end{equation}
This can be verified  from 
the normalization condition $i^{s,0} \in Z$, 
$i^{s',0} \in Y$ with $i^{s'} \in Z$ and $\size_{s'}(p_i) = 0$ for $s' \in S-s$.
For $s \in S$, since $\dist (0, p_s)=0$ implies $\chi_{Y,Z}(s^0) \leq 0$, 
a similar consideration (more easily) verifies 
\begin{equation}\label{eq:s_in_S_cost}
- r_s \dist (0, p'_s) + r_s \dist (0, p_s) = r_s (p'(s^0) - p(s^0)) = r_s \chi_{Y,Z}(s^0)/2 = - \kappa_{s^0s^0}(Y,Z)/2.
\end{equation}
By \eqref{eq:ij_cost}--\eqref{eq:s_in_S_cost}, we obtain (1).

(2) Let $(Y,Z)\in 3^U$ be a maximum violating cut.
We first modify $(Y,Z)$, without decreasing $\kappa$, so that $p+ \epsilon \tilde \chi_{Y,Z}$ for small $\epsilon >0$
 satisfies \eqref{eq:s_in_S}--\eqref{eq:i_in_V0}
 (except being half-integer-valued).
Namely  $p+ \epsilon \tilde \chi_{Y,Z}$ is viewed as a ``continuous'' subtree-potential.
For each $i \in V_s\ (s\in S)$, do the following:
\begin{itemize}
\item
If $\size (p_i) = 0$,  $Y\cap U_i=\emptyset$, and $Z\cap U_i\neq \emptyset$, then remove $U_i$ from $Z$.
\end{itemize}
This modification does not decrease $\kappa$,
 which can be seen from $\overline{c}(e_i) = 0$.
For $i \in V_0$,  do the following:
\begin{itemize}
\item For ``becoming $s$-type'' configuration
\[
i^{s,0} \in Z, \ i^{s',0} \in Y,\ i^{s'} \in Z,\  \size_{s'}(p_i) = 0 \quad (\exists s \in S, \forall s' \in S-s),
\]
remove $i^s$ from $Z$ and add to $Y$ if $\size_{s}(p_i) = 0$.
For other configuration,
remove $U_i^0$ from $Y \cup Z$ and remove each $i^{s}\ (s\in S)$ from $Z$ if $\size_{s}(p_i) = 0$.
\end{itemize}
This modification does not decrease $\kappa$.
This can be seen from observations: 
(i) $i^s$ can be removed from $Z$ if $i^{s,0} \not \in Y$ and $\size_s (p_i) = 0$ 
($\overline{c}(e_i^s) = 0$),
and can be added to $Y$ if moreover $i^{s,0}\in Z$.
(ii) $i^{s,0}$ can be removed from $Y$ 
if there is no $s' \neq s$ with $i^{s',0} \in Z$.
(iii) Suppose that $i^{s,0} \in Z$. 
Then $i^{s',0} \in Y$ for all $s' \in S-s$ by $\kappa(Y,Z)>0$.
If there is $s' \in S-s$ such that $\size_{s'}(p_i) > 0$ 
or $i^{s'} \not \in Z$, then $i^{s,0}$ and $i^{s',0}$ can be removed from $Z$ and $Y$, respectively;
the removal of $i^{s,0}$ decreases $\kappa$ by most $c_i$ 
but the removal of $i^{s',0}$ increases $\kappa$ by $c_i$.

After the procedure, 
we obtain a maximum (normalized) violating cut $(Y,Z)$ 
so that $p + \epsilon \tilde \chi_{Y,Z}$ is a continuous subtree-valued potential.
Still $p^{Y,Z}$ is not necessarily movable. 
This is caused by subtrees $p_i$ with $\size(p_i) = 1/2$ and 
$\chi_{Y,Z}(i^{s},i^{0}) = -2$ (or $\chi_{Y,Z}(i^s, i^{s'}) = -2$).
We decompose $(Y,Z)$ into upward- and downward-movable cuts. 
Let $(Y,Z)^\uparrow :=(Y\cap U^\uparrow, Z\cap U^\uparrow)$, and 
$(Y,Z)^\downarrow :=(Y\cap U^\downarrow,Z\cap U^\downarrow)$. 
They are movable.
Indeed, 
subtrees $p_i$ with $\size(p_i) = 1/2$ cannot become ``overturned,'' 
since $2 p(i^s)$ and $2 p(i^0)$ (or $2 p(i^{s'})=0$) have different parities.
The normalization condition obviously holds by construction.
Hence $(Y,Z)^\uparrow$ and $(Y,Z)^\downarrow$ 
are upward- and downward-movable, respectively. 

Moreover, the following holds:
\begin{equation}
 \label{eq:ntb_viocut_sep}
 \kappa(Y,Z)=\kappa((Y,Z)^\uparrow) +\kappa((Y,Z)^\downarrow ).
\end{equation}
Indeed, 
the last equality in \cref{eq:i_in_Vs_cost}--\cref{eq:s_in_S_cost}
holds also for $(Y,Z)$ 
(since they only use \cref{eq:s_in_S}--\cref{eq:i_in_V0}).
By $\tilde \chi_{Y,Z} = \tilde \chi_{(Y,Z)^\uparrow} + 
 \tilde \chi_{(Y,Z)^\downarrow}$,
  it holds $\kappa_{uv}(Y,Z)=\kappa_{uv} ((Y,Z)^\uparrow) +\kappa_{uv}((Y,Z)^\downarrow )$
 for these edges $uv$.
Similarly the last equality in \cref{eq:ij_cost} holds for $(Y,Z)$
 if $\dist(p_i,p_j)> a_{ij}$.
If $\dist(p_i,p_j)=a_{ij}$, then it holds
 $\dist(p_i,p_j) = - p(i^{\sigma}) - p(j^{\sigma'})\in \mathbb{Z}_+$.
Hence $2p(i^{\sigma})$ and $2p(j^{\sigma'})$ have the same parity.
Consequently $\kappa_{uv}(Y,Z)=\kappa_{uv} ((Y,Z)^\uparrow)$ 
with $\kappa_{uv}((Y,Z)^\downarrow ) = 0$ or
 $\kappa_{uv}(Y,Z)=\kappa_{uv} ((Y,Z)^\downarrow)$ 
with $\kappa_{uv}((Y,Z)^\uparrow ) = 0$. Summarizing, we have (\ref{eq:ntb_viocut_sep}).
By $\kappa(Y,Z)>0$, one of $(Y,Z)^\uparrow$  and $(Y,Z)^\downarrow$
is violating.

Consider arbitrary upward-  and downward-movable cuts 
$(Y',Z')$ and $(Y'',Z'')$, and consider the cut $(Y' \cup Y'', Z' \cup Z'')$. 
By the same argument, it holds $\kappa(Y' \cup Y'', Z' \cup Z'') = \kappa(Y',Z') + \kappa(Y'',Z'')$.
This implies the maximality of $(Y,Z)^\uparrow$  and $(Y,Z)^\downarrow$. 
\end{proof}

\begin{corollary}
\label{cor:bisub_nonopt}
Let $\mathcal{N}_p=((U,F),\underline{c},\overline{c})$
be the undirected network constructed from a proper potential
$p\in\mathbb{S}^n$.
Suppose that the instance is infeasible.
Given a maximum violating cut,
we can obtain a proper potential $q\in \mathbb{S}^n$ with $h(q)<h(p)$ in $O(kn)$ time.
\end{corollary}

\begin{proof}
By \cref{thm:map}~(2), we can obtain
 a maximum upward-movable cut $(Y,Z)\in \mathcal{M}^\uparrow$
and a maximum downward-movable cut $(Y',Z')\in \mathcal{M}^\downarrow$ in $O(kn)$ time.
Let $(Y^*,Z^*)$ be the cut attaining the maximum $\kappa$-value among $\{(Y,Z),(Y',Z')\}$,
and let $q:=p^{Y^*,Z^*}\in \mathbb{S}^n$.
Then $h(q)<h(p)$ by \cref{thm:map}~(1) and (2).
We can make $q$ proper by the procedure given in the beginning of the proof of \cref{prop:weakdual}.
\end{proof}

\subsection{Cut-Descent Algorithm}
\label{subsec:algo}

Now we are ready to present our algorithm \textsc{Descent} for
minimization of $h$; see \cref{algo:descent}.
The input is the network $((V,E),S,u,c,a,r)$,
and the output is a pair $(p,x)$ of optimal solutions for DTB and FNTB.
As mentioned in \cref{subsec:dual},
 the algorithm initializes $p\equiv 0$.
At the beginning of each iteration,
construct the network $\mathcal{N}_p$
 from the current proper potential $p\in \mathbb{S}^n$,
and run the algorithm given in \cref{lem:feascirc} to solve the undirected circulation problem.
There are two cases:
 the instance is feasible and a feasible half-integer-valued circulation is obtained,
or the instance is infeasible and a maximum violating cut is obtained.
If the former is the case,
then a half-integral optimal edge-capacity is obtained by \cref{thm:bisub_opt}.
If the latter is the case,
then a proper potential $q\in \mathbb{S}^n$ with $h(q)<h(p)$ is obtained by \cref{cor:bisub_nonopt}.
Each iteration of \textsc{Descent} can be done in $O(\MF(kn,m+k^2n))$ time.

\begin{algorithm}[t]
\caption{\textsc{Descent}}
\label{algo:descent}
\begin{enumerate}
\setcounter{enumi}{-1}
\item Initialize $p\equiv 0$ (i.e., $p(i)=0$ for any $i\in V$).
\item Construct $\mathcal{N}_p$ and solve the undirected circulation problem (by \cref{lem:feascirc}).
\item If a feasible half-integer-valued circulation is obtained,
	then $p$ is optimal;
	compute a half-integral optimal edge-capacity $x$ (by \cref{thm:bisub_opt})
	and output $(p,x)$.
\item If a maximum violating cut is obtained,
	then compute $q\in \mathbb{S}^n$ with $h(q)<h(p)$ (by \cref{cor:bisub_nonopt}),
	update $p\leftarrow q$, and go to Step 1.
\end{enumerate}
\end{algorithm}

The value $-h(p)$ is at most $mUA$ (by \cref{prop:weakdual})
and $-h(p)\in \mathbb{Z}_+^*$.
Thus the number of iterations is at most $O(mUA)$.
This is very rough analysis;
the number of iterations can be evaluated as $O(nA)$.
Recall $\mathbb{S}$ can be seen as the subset of $\mathbb{Z}^k$.
Let $\lVert T-T'\rVert:=\max_{s\in S}\lvert T_s-T'_s\rvert$
for $T,T'\in \mathbb{S}$.
For two potentials $p,q\in \mathbb{S}^n$, let $\lVert p-q\rVert:=\max_{i\in V} \lVert p_i-q_i\rVert$.
Let $\opt(h)$ denote the set of minimizers of $h$.

\begin{lemma}
\label{lem:descentnum}
With the initial potential $p_0\in \mathbb{S}^n$,
\textsc{Descent} finds an optimal potential at most $2\min_{q\in \opt(h)} \lVert p_0-q\rVert+2$ iterations.
\end{lemma}

\Cref{lem:descentnum} can be shown by using \emph{DCA beyond $\mathbb{Z}^n$},
which we will discuss in \cref{sec:disc}.

\begin{lemma}
\label{lem:range}
There exists an optimal potential $p\in \opt(h)$
such that, for any $i\in V$,
$p_i$ is contained in $[2nA,2nA,\dotsc,2nA]\in \mathbb{S}$.
\end{lemma}

\begin{proof}
By \cref{prop:weakdual},
there exists a proper optimal potential $p\in \mathbb{S}^n$.
Suppose that there is $s\in S$ such that $\dist(0,p_s)>2nA$.
Let $\mathcal{S}$ be the set of endpoints of $p_i$ located in $P_s$,
i.e., $\mathcal{S}:=\{p_s\}\cup
		\bigcup_{i\in V_s} \{(-p(i^0),s),(p(i^s),s)\}
		\cup \bigcup_{i\in V_0} \{(p(i^s),s)\}$.
Sort $\mathcal{S}=\{(l_1,s),(l_2,s),\dotsc,(l_N,s)\}$ in
the ascending order $l_0:=0\leq l_1<l_2<\dotsb<l_N$.
By $N\leq 2n$ and
$\sum_{t=1}^N (l_t-l_{t-1})=\dist(0,p_s)>2nA$.
there exists $t\in \{1,\dotsc,N\}$ such that $l_t-l_{t-1}>A$.

Let $e$ be the edge in the subtree $[l_{t-1},l_t]_s$ farthest from 0.
Consider the biset $\hat{X}_e:=(X_e,X^+_e)$ (defined in the proof of \cref{prop:weakdual}).
Then $u(\delta\hat{X}_e)+c(\Gamma(\hat{X}_e))\geq r(s)$ by \eqref{cond:ntbfeas}.
We define a new potential $p'\in \mathbb{S}^n$ by ``shrinking'' $e$.
Formally, let $Y:=\{i^s\mid i\in X_e^+\}$ and $Z:=\{i^0\mid i\in X_e\}$.
Then $p'\in \mathbb{S}$ is defined by $p':=p-(1/2)\chi_{Y,Z}$.
Observe that $\size(p_i)$ for $i\in \Gamma(\hat{X}_e)$ decreases by $1/2$,
and $\dist(p_i,p_j)\,(> l_t-l_{t-1}>A\geq a_{ij})$ for $ij\in \delta(\hat{X}_e)$
 also decreases by $1/2$.
Hence we have
\[
    h(p')-h(p)=\frac{1}{2}(r_s-c(\Gamma(\hat{X}_e))-u(\delta\hat{X_e}))\leq 0.
\]
Thus $p'$ is also an optimal potential.
We apply this procedure iteratively unless the optimal potential does not contained in $[2nA,\dotsc,2nA]$.
(It terminates since $\sum_{s\in S}\dist(0,p_s)$ is strictly decreasing.)
\end{proof}

\begin{corollary}
\label{cor:pseudo}
\textsc{Descent} solves FNTB in $O(nA\cdot \MF(kn,m+k^2n))$ time.
\end{corollary}

We prove \cref{thm:lc} using \cref{cor:pseudo}.

\begin{proof}[Proof of \cref{thm:lc}]
Let $((V,E),S,u,c)$ be a network,
and let $f=(\mathcal{P},\lambda)$ be a separately-capacitated multiflow.
Recall that $f_s=(\mathcal{P}_s,\lambda|_{\mathcal{P}_s})$,
where $\mathcal{P}_s\subseteq\mathcal{P}$ is the subset of paths connecting $s$ to $S-s$.
Let $\val f:=\sum_{P\in\mathcal{P}} \lambda(P)$
and $\val f_s:=\sum_{P\in\mathcal{P}_s} \lambda(P)\ (s\in S)$.
Since $f_s$ can be seen as an ordinary $\{s\}$--$(S-s)$ flow
and is capacitated by $u$ and $c$,
 $\val f_s$ is at most the minimum capacity of $\{s\}$--$(S-s)$ cuts.
Thus we have $\val f=(1/2)\sum_{s\in S} \val f_s\leq (1/2)\sum_{s\in S}\nu_s$.

Consider an instance $I=((V,E),S,u,c,a,r)$ of FNTB,
where $a\equiv 1$ and $r_s:=\nu_s$ for each $s\in S$.
Since $u$ clearly satisfies \eqref{cond:ntbfeas}, the instance is feasible.
Then \textsc{Descent} outputs a half-integral optimal edge-capacity $x$
(and an optimal potential) for $I$.
Since $x$ comes from a half-integral circulation $y$ (\cref{thm:bisub_opt}),
it holds $x(\delta i)\in \mathbb{Z}_+$ for any $i\in V\setminus S$.
(For $i\in V_s\ (s\in S)$, this follows from $x(\delta i)=-2y(i^0i^s)$ and $y(i^0i^s)\in \mathbb{Z}^*$.
For $i\in V_0$, this follows from $x(\delta i)=\sum_{s\in S}-y(i^{s,0}i^{s})=2\sum_{s<s'}y(i^{s,0}i^{s',0})$
and $y(i^{s,0}i^{s',0})\in \mathbb{Z}^*$.)
Then we can apply \textsc{Decompose} to $x$,
and by \Cref{cor:halfflow},
a half-integer-valued separately-capacitated multiflow $f$ is obtained.
It satisfies $\val f=(1/2)\sum_{s\in S}f(s)\geq
(1/2)\sum_{s\in S}r_s=(1/2)\sum_{s\in S}\nu_s$.

The time complexity $O(n\cdot\MF(kn,m+k^2n))$ follows from \cref{cor:pseudo,cor:halfflow}.
\end{proof}

\subsection{Scaling Algorithm}
\label{subsec:scaling}

The time complexity of \textsc{Descent} is pseudo-polynomial.
We improve it by combining with the (cost-)scaling method.

Let $\gamma\in\mathbb{Z}_+$ be an integer
such that $2^\gamma\geq A$.
The scaling algorithm consists of $\gamma+1$ phases.
In $t$-th phase, solve DTB with an edge-cost $a_t:E\rightarrow\mathbb{Z}_+$
defined by $a_t(e):=\lceil a(e)/2^t \rceil\ (e\in E)$,
in other words, minimize $h_{a_t}$. (Recall that $h_a$ is defined by \eqref{eq:horig}.)
Here $\lceil \cdot\rceil$ is the round-up operator.
Note that all $a_t(e)$ are positive.
Begin with $t=\gamma$, and decrease $t$ one-by-one.
Then, when $t=0$, the problem coincides with the original DTB.
In each $t$-th phase, we use \textsc{Descent} to minimize $h_{a_t}$.
At the initial phase $t=\gamma$, we run \textsc{Descent} with the starting point $p\equiv 0$.
For $t$-th phase with $t\leq \gamma-1$, the starting point
is determined from the optimal potential in the previous phase.
The following lemma is easily seen from the strong duality of \cref{prop:weakdual}.

\begin{lemma}
\label{lem:double}
Let $p\in\mathbb{S}^n$ be a minimizer of $h_{a_t}\ (t=1,\dotsc,\gamma)$.
Then $2p\in\mathbb{S}^n$ is a minimizer of $h_{2a_t}$.
\end{lemma}

Here, $2p$ is naturally defined via the identification between
$\mathbb{S}^n$ and (the subset of) $(\mathbb{Z}^*)^{kn}$.
Observe that $a_{t-1}=2a_t-\sum\{\chi_e\mid e\in E,\ \text{$a_{t-1}(e)$ is odd}\}$.
The key property is the following sensitivity result.

\begin{lemma}
\label{lem:sensitivity}
Let $a:E\rightarrow\mathbb{Z}_+$ be a positive edge-cost.
Let $e\in E$ be an edge satisfying $a(e)\geq 2$,
and $a':=a-\chi_{e}$.
Let $p\in \opt(h_a)$.
Then $\min_{q\in\opt(h_{a'})}\lVert p-q\rVert\leq 1$.
\end{lemma}

We prove \cref{lem:sensitivity}
in \cref{sec:sensitivity} using DCA beyond $\mathbb{Z}^n$.

\begin{proof}[Proof of \cref{thm:main}]
For the initial phase $t=\gamma$, an optimal potential can be obtained in $O(n)$ iterations
of \textsc{Descent} by \cref{cor:pseudo}.
For each subsequent scaling phase, an optimal potential can be obtained in
$O(m)$ iterations of \textsc{Descent} by \cref{lem:descentnum,lem:double,lem:sensitivity}.
Let us analyze the time complexity.
Since we first do the reduction for \textbf{(CP)} (\cref{subsec:red}),
the number of phases is $O(\log(mUA))=O(\log(nUA))$.
Thus the number of iterations of \textsc{Descent} is $O(m+m\log(nUA))=O(m\log(nUA))$,
and the overall time complexity is $O(m\log(nUA)\cdot \MF(kn,m+k^2 n))$.
\end{proof}

\section{Analysis via DCA beyond $\mathbb{Z}^n$}
\label{sec:disc}

In this section, we analyze the cost-scaling algorithm
in terms of Discrete Convex Analysis (DCA) beyond $\mathbb{Z}^n$.
We show that \textsc{Descent} is viewed as the \emph{steepest descent algorithm (SDA)}
for an L-convex function on a median graph constructed from the subtree space,
and prove \cref{lem:descentnum} from a general result of SDA.
Also we introduce a new discrete convexity concept, \emph{N-convexity},
to prove \cref{lem:sensitivity}.

\subsection{L-convex Functions on Median Graphs}

We briefly introduce a theory of
discrete convexity on graph structures specialized to median graphs.
See \cite{Chepoi2000Graphs,Hirai2016Discrete,Hirai2018L} for more general theory.
We use basic terminologies of poset (partially ordered set) and lattice.

Let $G$ be a (possibly infinite) undirected graph.
We denote the set of nodes also by $G$.
Let $d=d_G$ be the shortest path metric on $G$.
For $u,v\in G$, the \emph{metric interval} $I(u,v)$ of $u,v$ is defined as
 the set of $w\in G$ satisfying $d(u,v)=d(u,w)+d(w,v)$.
A \emph{median graph} $G$ is a graph that for any $u,v,w\in G$,
$I(u,v)\cap I(v,w)\cap I(w,u)$ is a singleton.

We consider an \emph{orientation} on edges of a median graph $G$,
that takes either $u\searrow v$ or $u\swarrow v$ on each edge $uv$.
An orientation is \emph{admissible}
if for any 4-cycle $(u_1,u_2,u_3,u_4)$,
$u_1\searrow u_2$ implies $u_4\searrow u_3$.
An \emph{oriented} median graph is a median graph endowed with an admissible orientation.
It is known~\cite[Lemma 2.4]{Hirai2016Discrete} that
an admissible orientation is acyclic.
By considering the transitive (and reflective) closure of the orientation,
$G$ is viewed as a poset with a partial order $\preceq$,
where $u\swarrow v\Rightarrow u\preceq v$.
For $u\preceq v$, the \emph{interval} $[u,v]$
is the set of $w\in G$ satisfying $u\preceq w\preceq v$.
It is known (see, e.g., \cite{Vel1993Theory}) that
every interval of an oriented median graph $G$
 is a distributive lattice.
An oriented median graph $G$ is \emph{well-oriented}
if every interval is a Boolean lattice.
The (Cartesian) product of two well-oriented median graphs
becomes a well-oriented median graph with a natural orientation.

Let $G$ be a well-oriented median graph.
A pair of nodes $(u,v)$ of $G$ is \emph{non-comparable}
 if $u\not\preceq v$ and $u\not\succeq v$.
Also the pair $(u,v)$ is \emph{bounded above (resp. bounded below)}
 if there exists $w\in G$ such that $u\preceq w\succeq v$ (resp. $u\succeq w\preceq v$).
We call such a pair \emph{upper-bounded} (resp. \emph{lower-bounded}).
For an upper-bounded (resp. lower-bounded) pair $(u,v)$,
 there always exists the join $u\vee v$ (resp. the meet $u\wedge v$) of $u$ and $v$.
A pair bounded above and below is called \emph{bounded}.
A sequence $(u=u_0,u_1,\dotsc,u_\ell=v)\subseteq G$ is a \emph{chain
from $u$ to $v$ with length $\ell$} if $u_{i-1}\preceq u_i$ for any $i=1,\dotsc,\ell$.
For $u\preceq v$ (or $v\preceq u$),
 let $r[u,v]$ be the maximum length of a chain from $u$ to $v$ (or from $v$ to $u$).
A non-comparable lower-bounded pair $(u,v)$ of $G$
is \emph{$\wedge$-antipodal} if
every upper-bounded pair $(a,b)$ of $G$ with
 $u\succeq a\succeq u\wedge v\preceq b\preceq v$
satisfies
\begin{equation}
\label{eq:anti_def}
	r[a,u]r[b,v]\geq r[u\wedge v,a]r[u\wedge v,b].
\end{equation}
This means that a point $(r[u\wedge v,a],r[u\wedge v,b])\in \mathbb{R}^2$
is under or lying on the segment connecting $(r[u\wedge v,u],0)$ and $(0,r[v\wedge v,v])$
(since $r[u\wedge v,u]=r[u\wedge v,a]+r[a,u]$ and so is $r[u\wedge v,v]$).
By the definition, a $\wedge$-antipodal pair is not (upper-)bounded.
Similarly, a non-comparable upper-bounded pair $(u,v)$ of $G$
is \emph{$\vee$-antipodal} if every lower-bounded pair $(a,b)$ of $G$ with
$u\preceq a\preceq u\vee v\succeq b\succeq v$ satisfies
\cref{eq:anti_def} with replacing $\wedge$ by $\vee$.
An \emph{antipodal} pair of $G$ is a pair $\wedge$-antipodal or $\vee$-antipodal.

A function $f:G\rightarrow\mathbb{R}$ is called \emph{L-convex}~\cite[Proposition~8]{Hirai2018L}
 if it satisfies the following conditions:
\begin{itemize}
\item For every bounded pair $(u,v)$ of $G$,
 it satisfies the \emph{submodularity inequality}
\begin{equation}
\label{eq:submod}
    f(u)+f(v)\geq f(u\wedge v)+f(u\vee v).
\end{equation}
\item For every $\wedge$-antipodal pair $(u,v)$ of $G$,
 it satisfies the \emph{$\wedge$-convexity inequality}
\begin{equation}
\label{eq:anti}
    r[u\wedge v,v] f(u)+ r[u\wedge v,u] f(v)
        \geq (r[u\wedge v,u]+r[u\wedge v,v])f(u\wedge v).
\end{equation}
Equivalently, the point $(0,f(u\wedge v))\in\mathbb{R}^2$ is under or lying on the segment
between $(-r[u\wedge v,u],f(u))$ and $(r[u\wedge v,v],f(v))$.
\item For every $\vee$-antipodal pair $u,v$ of $G$,
it satisfies \cref{eq:anti} with replacing $\wedge$ by $\vee$.
\end{itemize}

The global optimality of an L-convex function $f$
is characterized by a \emph{local} condition.
For $u\in G$, the \emph{principal ideal}  $\mathcal{I}_u$ 
(resp. \emph{principal filter} $\mathcal{F}_u$)
is the set of $v\in G$ satisfying $v\preceq u$ (resp. $v\succeq u$).
$u\in G$ is a minimizer of $f$ if and only if
$u$ is a minimizer of $f$ restricted to $\mathcal{F}_u\cup \mathcal{I}_u$~\cite[Theorem~4.2]{Hirai2018L}.
This fact suggests a natural minimization algorithm,
called the \emph{steepest descent algorithm} (SDA); see \cref{algo:SDA}.
The algorithm iteratively updates the current solution $u$
to a local minimizer around $u$.
Then the output of SDA (if it terminates) is optimal.
The number of iterations is bounded
by the \emph{$\Delta$-distance} between the initial point and $\opt(f)$.
Here, a \emph{$\Delta$-path} $(u=u_0,u_1,u_2,\dotsc,u_\ell=v)$
 from $u$ to $v$ is a sequence of $G$ such that
for any $1\leq i\leq \ell$, there exists a cube subgraph containing $\{u_{i-1},u_i\}$,
and the \emph{$\Delta$-distance} $d^\Delta_G(u,v)=d^\Delta(u,v)\ (u,v\in G)$ is
 the minimum length of $\Delta$-paths from $u$ to $v$.

\begin{algorithm}[t]
\caption{SDA}
\label{algo:SDA}
\begin{enumerate}
\setcounter{enumi}{-1}
\item Take any $u\in G$.
\item Find a minimizer $v$ of $f$ over $\mathcal{F}_u\cup\mathcal{I}_u$.
\item If $f(v)=f(u)$, then stop; output $u$. Otherwise update $u$ to $v$
    and go to Step 1.
\end{enumerate}
\end{algorithm}

\begin{theorem}[{\cite[Theorem 4.3]{Hirai2018L}}]
\label{thm:iterSDA}
The number of iterations of SDA with an initial point $u\in G$
is at most $\min_{v\in\opt(f)}d^\Delta(u,v)+2$.
\end{theorem}

The half-integer lattice $(\mathbb{Z}^*)^k$ is regarded as an oriented grid graph,
where two half-integers $x,y$ are joined by an edge
 if $x-y=\pm(1/2)\chi_i$ for some $i$,
with the orientation given by $x\swarrow y$ if $x_i\in \mathbb{Z}$
(equivalently, $y_i\in \mathbb{Z}^*\setminus\mathbb{Z}$).
Then $(\mathbb{Z}^*)^k$ is a well-oriented median graph.
Also $(\mathbb{Z}^*_+)^k$, the induced subgraph of $(\mathbb{Z}^*)^k$
to the nonnegative orthant,
is a well-oriented median graph.

The following simple L-convex functions on median graph $(\mathbb{Z}^*)^k$
are building blocks
to represent the objective function of DTB
as an L-convex function on the median graph of the subtree space.

\begin{lemma}
\label{lem:Lconv_grid}
\begin{enumerate}
\renewcommand{\labelenumi}{\textup{(\arabic{enumi})}}
\item $(\mathbb{Z}^*)^k\ni x\mapsto a^\top x+b$ is L-convex for $a\in \mathbb{R}^k$ and $b\in\mathbb{R}$.
\item $(\mathbb{Z}^*)^k\ni x\mapsto (x_i+x_j-a)^+$ is L-convex for $i,j\in\{1,\dotsc,k\}$ and $a\in\mathbb{Z}$.
\item The L-convexity is preserved under an individual sign inversion of a variable,
i.e., if a function $f:(\mathbb{Z}^*)^k\to \mathbb{R}$ is L-convex,
 then $x\mapsto f(x_1,\dotsc,-x_i,\dotsc,x_k)$ is also L-convex.
\end{enumerate}
\end{lemma}

\begin{proof}
For an antipodal pair $(x,y)$ of $(\mathbb{Z}^*)^k$,
at least one of indices $i$ satisfies $\lvert x_i-y_i\rvert = 1$,
and $x_j=y_j$ for any other index $j$.
Indeed, if $(x,y)$ is $\wedge$-antipodal and
 there exists $j$ with $\lvert x_j-y_j\rvert > 1$,
then $x,y$ cannot be bounded below.
If there exists $j$ with $\lvert x_j-y_j\rvert=1/2$,
then the following $(a,b)$ violates \cref{eq:anti_def}:
We can assume $y_j\in \mathbb{Z}^*\setminus\mathbb{Z}$ without loss of generality.
Let $a:=x$ and let $b\in (\mathbb{Z}^*)^k$ defined by
$b_j:=y_j$ and $b_i:=(x\wedge y)_i$ for any other $i\neq j$.
Then $(a,b)$ is bounded above and satisfies
 $x\succeq a\succeq x\wedge y\preceq b\preceq y$ but violates \cref{eq:anti_def}.
For a $\vee$-antipodal pair, the argument is the same.
Then $x\wedge y$ (or $x\vee y$) is equal to $(1/2)(x+y)$
and the $\wedge$($\vee$)-convexity inequality \cref{eq:anti}
follows from the ordinary convexity inequality.
Thus (1) and (2) satisfy the $\wedge$($\vee$)-convexity inequality \cref{eq:anti}.

If $x,y\in (\mathbb{Z}^*)^k$ is bounded,
then $\lvert x_i-y_i\rvert\leq 1/2$ for any $i$.
Thus the submodularity inequality \cref{eq:submod} for (1)
 is obvious from its linearity.
For (2), it is sufficient to show
 the submodularity inequality \cref{eq:submod}
for $f':\{-1/2,0,1/2\}^2\to \mathbb{Z}^*$ defined by
$f'(x_1,x_2):=(x_1+x_2-a)^+\ (a\in \mathbb{Z})$.
If $a\leq -1$, then $f'$ becomes a linear function $x_1+x_2-a$.
If $a\geq 1$, then $f'\equiv 0$.
So we may assume $a=0$.
Then the submodularity inequality \cref{eq:submod} is easily verified, e.g.,
$f'(-1/2,0)+f'(0,1/2)=0+1/2>0+0=f'(-1/2,1/2)+f'(0,0)$.

(3) The boundedness of a pair $(x,y)$ of $(\mathbb{Z}^*)^k$
and the antipodality are preserved
 under an individual sign inversion.
Also the sign inversion commutes with $\wedge$ and $\vee$,
i.e., the sign inversion of $x\wedge y$ (resp. $x\vee y$) is
the meet (resp. join) of the sign inversions of $x$ and $y$.
Thus the statement holds.
\end{proof}

We note this L-convexity is not equivalent to
 the original L(${}^\natural$)-convexity in DCA~\cite{Murota2003Discrete}.

The following criteria is useful for checking L-convexity of
 a function defined on
the product of two median graphs.
A bounded pair $(x,y)$ of $G$ is \emph{$2$-bounded}
if there is no $z\in G$ satisfying
$x\prec z\prec x\vee y$ or $y\prec z\prec x\vee y$.

\begin{lemma}[{\cite[Proposition 3.8]{Hirai2016Discrete}}]
\label{lem:Lconv_two}
Let $G$ and $G'$ be median graphs.
A function $f:G\times G'\to \mathbb{R}$ is L-convex
 if and only if it satisfies
\begin{enumerate}
\renewcommand{\labelenumi}{\textup{(\arabic{enumi})}}
\item the submodularity inequality \eqref{eq:submod} for every 2-bounded pair $(x,y)$, and
\item the $\wedge$($\vee$)-convexity inequality \eqref{eq:anti} for every pair
    $(x,y)=((x_1,x_2),(y_1,y_2))$
    such that $x_1=y_1$ and $(x_2,y_2)$ is antipodal in $G'$,
    or $x_2=y_2$ and $(x_1,y_1)$ is antipodal in $G$.
\end{enumerate}
\end{lemma}

\subsection{L-convexity of DTB}
\label{subsec:LconvDTB}

Let $\mathbb{C},\mathbb{W}_s\ (s\in S)$ be subsets of $(\mathbb{Z}^*)^k$
 defined by
\begin{align*}
	\mathbb{C}&:=\{T\in(\mathbb{Z}^*)^k\mid T_s\geq 0\ (\forall s\in S)\},\\
	\mathbb{W}_s&:=\{T\in(\mathbb{Z}^*)^k\mid T_s\geq 0,\ T_t=T_{t'}\leq 0\ 
		(\forall t,t'\in S-s)\}.
\end{align*}
Let $\overline{\mathbb{S}}:=\mathbb{C}\cup\bigcup_{s\in S}\mathbb{W}_s$.
We endow $\overline{\mathbb{S}}$ with a graph structure as follows.
Now $\mathbb{C}$ is naturally regarded as an oriented grid graph $(\mathbb{Z}^*_+)^k$.
Also each $\mathbb{W}_s$ is viewed as an oriented grid graph $\mathbb{Z}^*_+\times (-\mathbb{Z}^*_+)$
via bijection $\mathbb{W}_s\ni T\mapsto (T_s,T_t)\in \mathbb{Z}^*_+\times (-\mathbb{Z}^*_+)$,
where $t\in S-s$.
This construction is well-defined,
since both constructions give the same (induced) subgraph on $\mathbb{C}\cap \mathbb{W}_s$.
This is nothing but the \emph{gated amalgam} of
$(\mathbb{Z}^*_+)^k$ and several $\mathbb{Z}^*_+\times (-\mathbb{Z}^*_+)$;
see \cref{fig:ex_subtree}.
It is a folklore that the gated amalgam of median graphs is a median graph;
see, e.g., \cite{Vel1993Theory}.
Observe that the well-orientedness keeps.
Thus we have

\begin{figure}
\centering
\includegraphics[width=11cm]{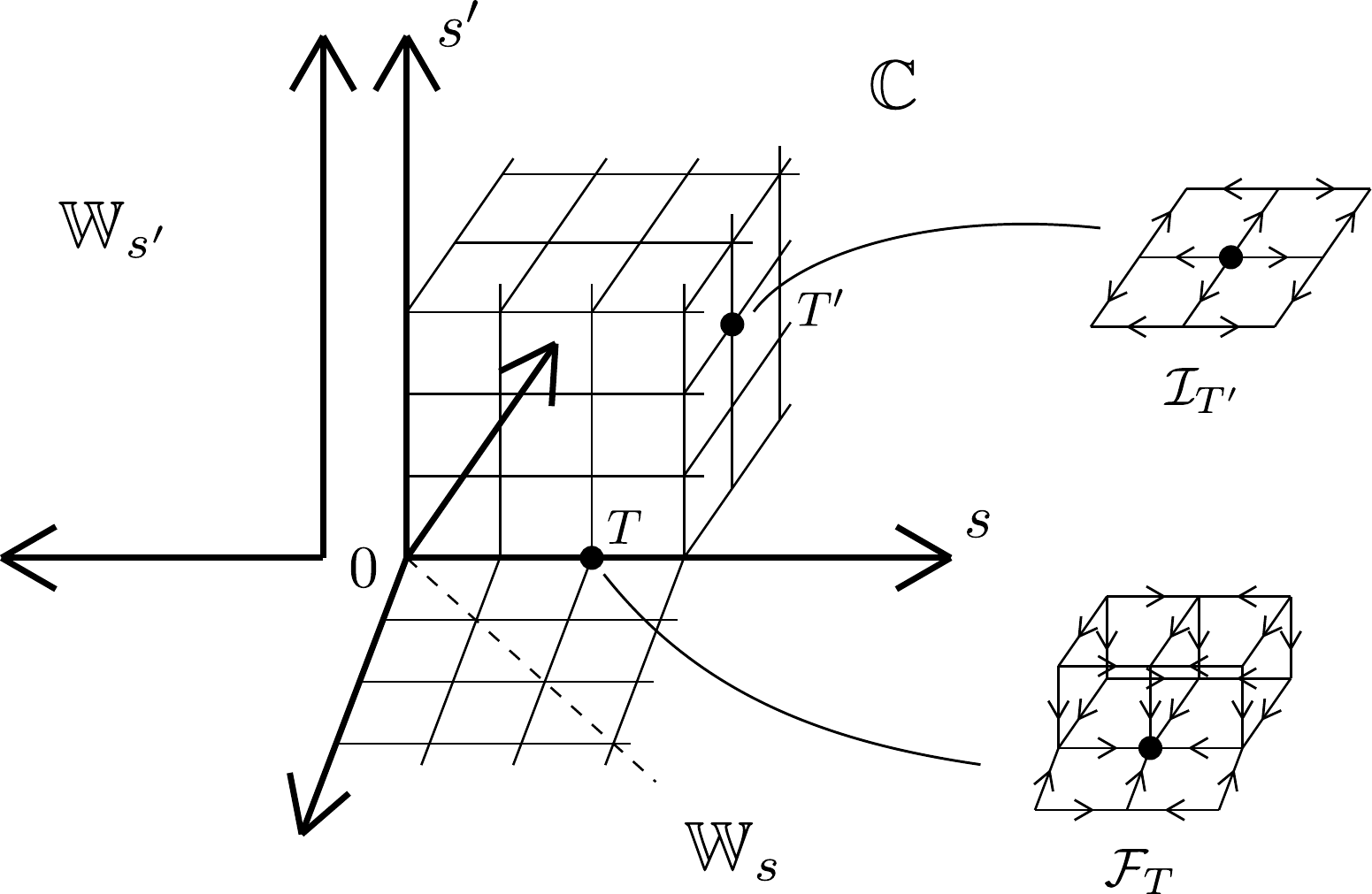}
\caption{The structure of $\overline{\mathbb{S}}$ ($\lvert S\rvert =3$).}
\label{fig:ex_subtree}
\end{figure}

\begin{proposition}
\label{prop:medgraph}
$\overline{\mathbb{S}}$ is a well-oriented median graph, and so is $\overline{\mathbb{S}}^n$.
\end{proposition}

(One can directly prove the medianness of $\overline{\mathbb{S}}$ from the definition.)

In this subsection, we show the function $h$ defined in \eqref{eq:horig}
 can be extended to an L-convex function
 on the (Cartesian) product of $\overline{\mathbb{S}}$,
and DTB is minimization of the L-convex function.
For representing \cref{cond:dtbfeas},
let $\mathbb{G}:=\prod_{s\in S}\mathbb{W}_s\times \overline{\mathbb{S}}^{n-k}
\subseteq \overline{\mathbb{S}}^n$,
which is a well-oriented median graph.
We first introduce a barrier function $B:\overline{\mathbb{S}}\to \mathbb{R}_+$
with a property that $B(T)=0$ if and only if $T\in \mathbb{S}\subseteq\overline{\mathbb{S}}$.
Let define $B:\overline{\mathbb{S}}\to \mathbb{R}_+$ by
\[
	B(T):=\begin{dcases*}
		0 & if $T\in \mathbb{C}$,\\
		(-T_s-T_t)^+ & if $T\in \mathbb{W}_s\ (s\in S)$,
	\end{dcases*}
\]
where $t\in S-s$ (arbitrary).
Next let $\size:\overline{\mathbb{S}}\to \mathbb{R}_+$ be
 a function on $\overline{\mathbb{S}}$ defined by
\[
	\size(T):=\begin{dcases*}
		\sum_{s\in S}T_s & if $T\in \mathbb{C}$,\\
		(T_s+T_t)^+ & if $T\in \mathbb{W}_s\ (s\in S)$,
	\end{dcases*}
\]
where $t\in S-s$ (arbitrary).
Notice that $\size(\cdot)$ is consistent with
 the definition given in \cref{subsec:space_of} on $\mathbb{S}$.
Let define $\dist_a:\overline{\mathbb{S}}^2\to \mathbb{R}_+
\ (0<a\in \mathbb{Z}_+)$ by
\[
	\dist_a(T,T'):=\begin{dcases*}
		0 & if $T,T'\in \mathbb{C}$,\\
		(-T_t-T'_s-a)^+ & if $T\in \mathbb{W}_s,\ T'\in \mathbb{C}\cup\mathbb{W}_t\ (s\neq t)$,\\
		(-T_s-T'_t-a)^+ + (-T_t-T'_s-a)^+ & if $T,T'\in \mathbb{W}_s$,\\
	\end{dcases*}
\]
where $t\in S-s$ (arbitrary) for the case $T,T'\in \mathbb{W}_s$.
$\dist_a(T,T')$ is equal to $(\dist(T,T')-a)^+$ on $\mathbb{S}^2$.
So we denote $\dist_0(\cdot,\cdot)$ by $\dist(\cdot,\cdot)$.
Let $M$ be a large number.
Define $\tilde{h}:\mathbb{G}\to \mathbb{R}$ by
\begin{align}
\label{eq:hL}
	\tilde{h}(p)&:=-\sum_{s\in S}r_s\dist(0,p_s)+\sum_{i\in V\setminus S}c_i\size(p_i)
	+\sum_{ij\in E}u_{ij}\dist_{a_{ij}}(p_i,p_j)\notag\\
	&\qquad +M\left(\sum_{s\in S}\left(B(p_s)+\size(p_s)\right)
		+\sum_{i\in V\setminus S}B(p_i)\right).
\end{align}
Since $M$ is large,
any minimizer $p$ of $\tilde{h}$ satisfies $p\in \mathbb{S}^n$
and $\size(p_s)=0\ (s\in S)$.
The latter condition is equivalent to $p_s\in P_s$.
Thus $p$ is a potential in the sense of \cref{cond:dtbfeas}.
Let define $\mathbb{D}\subseteq \mathbb{G}$ by
\[
	\mathbb{D}:=\{p\in \mathbb{G}\mid p_s\in P_s\ (s\in S),\ p_i\in \mathbb{S}\ (i\in V\setminus S)\}.
\]
Then \cref{eq:hL} is consistent with \cref{eq:horig} on $\mathbb{D}$.

\begin{proposition}
\label{prop:Lconv}
$\tilde{h}:\mathbb{G}\to \mathbb{R}$ is an L-convex function.
\end{proposition}

We will prove \cref{prop:Lconv} later.
By \cref{prop:Lconv}, we can apply SDA for minimizing $\tilde{h}$.
In fact, the algorithm \textsc{Descent} in the previous section
is precisely an implementation of SDA for $\tilde{h}$.
Recall \cref{subsec:direction} for
 the sets $\mathcal{M}^\uparrow$ and $\mathcal{M}^\downarrow$ 
of all upward-movable and downward-movable cuts in $\mathcal{N}_p$, respectively.
Then if $(Y,Z)\in \mathcal{M}^\uparrow$,
then $p_i\preceq p^{Y,Z}_i$ for any $i\in V$ and thus $p^{Y,Z}\in \mathcal{F}_p$
(and $p^{Y,Z}\in \mathbb{D}$).
Conversely, any $q\in \mathcal{F}_p\cap \mathbb{D}$
is written as $q=p^{Y,Z}$ with $(Y,Z)\in \mathcal{M}^\uparrow$.
Similarly, $\mathcal{M}^\downarrow$ corresponds to $\mathcal{I}_p\cap \mathbb{D}$.
Hence we have the following lemma.

\begin{lemma}
\label{lem:bijection}
The map $(Y,Z)\mapsto p^{Y,Z}$ is a bijection
from $\mathcal{M}^\uparrow$ (resp. $\mathcal{M}^\downarrow$)
 to $\mathcal{F}_p\cap D$ (resp. $\mathcal{I}_p\cap D$).
\end{lemma}

Thus, \cref{lem:bijection} and \cref{thm:map} say that
$p^{Y^*,Z^*}$ in the proof of \cref{cor:bisub_nonopt}
 is a local minimizer of $\tilde{h}$ at $p$ in the sense of SDA.
Now \cref{lem:descentnum} follows from \cref{thm:iterSDA}
and the following fact.

\begin{proposition}
\label{prop:distequiv}
For $p,q\in \mathbb{D}$, $d_\mathbb{G}^\Delta(p,q)=2\lVert p-q\rVert$.
\end{proposition}

\begin{proof}
It is sufficient to show $d_\mathbb{S}^\Delta(T,T')=2\lVert T-T'\rVert$ for any $T,T'\in\mathbb{S}$.
By the construction, we can easily see that $d_\mathbb{S}^\Delta(T,T')\geq 2\lVert T-T'\rVert$.
We will construct in \cref{sec:sensitivity}
a canonical $\Delta$-path with the length $2\lVert T-T'\rVert$
between any $T,T'\in\mathbb{S}$.
\end{proof}

The remaining part is proving \cref{prop:Lconv}.
It is known that the L-convexity is preserved under the addition
and the multiplication of a nonnegative real number~\cite[Lemma 4.9]{Hirai2016Discrete}.
Thus it is sufficient to show the following proposition.

\begin{proposition}
\label{prop:Lconv_term}
\textup{(1)} $-\dist(0,\cdot)$ on $\mathbb{W}_s\ (s\in S)$,
\textup{(2)} $B(\cdot)$ on $\overline{\mathbb{S}}$,
\textup{(3)} $\size(\cdot)$ on $\overline{\mathbb{S}}$,
\begin{align*}
	\text{\textup{(4)}}&\quad \overline{\mathbb{S}}\ni T\mapsto
	 \dist_a(T,R)+M\cdot B(T)\quad (0<a\in \mathbb{Z}_+,\ R\in\overline{\mathbb{S}}),\\
	\text{\textup{(5)}}&\quad \overline{\mathbb{S}}^2\ni (T,R)\mapsto
	 \dist_a(T,R)+M\cdot (B(T)+B(R))\quad (0<a\in \mathbb{Z}_+)
\end{align*}
are all L-convex, where $M$ is a large number.
\end{proposition}

\begin{proof}
As we have already noted,
 $\mathbb{W}_s=\mathbb{Z}^*_+\times(-\mathbb{Z}^*_+)$.
Thus the submodularity for (1) is obtained from \cref{lem:Lconv_grid}.

We show the submodularity for (2), (3) and (4).
Any bounded pair is contained in $\mathbb{C}$ or some $\mathbb{W}_s$.
Thus the submodularity inequality \cref{eq:submod} is obtained from \cref{lem:Lconv_grid}.
Similarly, for any antipodal pair
 contained in $\mathbb{C}$ or some $\mathbb{W}_s$,
the $\wedge$($\vee$)-convexity inequality \cref{eq:anti} follows from \cref{lem:Lconv_grid}.
Thus we can only consider a $\wedge$-antipodal pair $(T,T')$
with $T\wedge T'\in \mathbb{C}\cap \mathbb{W}_s\ (s\in S)$.
By $B(T\wedge T')=0$ for $T\wedge T'\in \mathbb{C}\cap\mathbb{W}_s$ and
 the nonnegativity of $B(\cdot)$, we may only consider (3) and (4).

Suppose that $T\wedge T'\in \mathbb{C}\cap \mathbb{W}_s\setminus\{0\}$.
We may assume $T\in \mathbb{W}_s\setminus\mathbb{C}$
and $T'\in\mathbb{C}\setminus\mathbb{W}_s$.
As \cref{lem:Lconv_grid},
we observe $\lvert T_s-T'_s\rvert\in\{0,1\}$,
$T_t=-1/2,\ T'_t\in\{0,1/2\}$ for any $t\in S-s$,
and at least one of $t\in S-s$ satisfies $T'_t=1/2$.
If $\lvert T_s-T'_s\rvert=1$, then exactly one $t\in S-s$ satisfies $T'_t=1/2$.
Otherwise, we can obtain $a,b$ violating \cref{eq:anti_def} as follows:
Let $a\in \overline{\mathbb{S}}$ defined by
$a_s:=T_s$ and $a_t:=0$ for any other $t\neq s$.
Let $b\in \overline{\mathbb{S}}$ defined by
$b_s:=(T_s+T'_s)/2$ and $b_t:=T'_t$ for any other $t\neq s$.
Then $(a,b)$ is bounded below and satisfies $T\succeq a\succeq T\wedge T'\preceq b\preceq T'$
but violates \cref{eq:anti_def} by $r[T\wedge T',b]>r[b,T']\ (=1)$
 and $r[T\wedge T',a]=r[a,T]\ (=1)$.

We check (the equivalent condition of)
 the $\wedge$-convexity inequality \cref{eq:anti} for (3) and (4).
Let $b:=(T\wedge T')_s-T_s\ (=T'_s-(T\wedge T')_s)\in\{-1/2,0,1/2\}$ and
$c:=\sum_{t\in S-s} T'_t\geq 1/2$.
For (3), we see
 $(-r[T\wedge T',T],\size(T))=(-2(1/2+\lvert b\rvert),\size(T\wedge T')-1/2-b)$
and $(r[T\wedge T',T'],\size(T'))=(2(c+\lvert b\rvert),\size(T\wedge T')+c+b)$.
We can determine that the point $(0,\size(T\wedge T'))$ is under or lying on the segment
connecting those two points in $\mathbb{R}^2$.

For (4), we note $B(T)=B(T')=B(T\wedge T')=0$.
We have three cases: (i) $R\in\mathbb{C}$,
(ii) $R\in\mathbb{W}_s\setminus\mathbb{C}$,
 (iii) $R\in\mathbb{W}_t\setminus\mathbb{C}\ (t\in S-s)$.
For (i), we have \cref{eq:anti} by $\dist_a(T\wedge T',R)+M\cdot B(T\wedge T')=0$
and the nonnegativity of the function.
For (ii), we can assume $\dist_a(T\wedge T',R)>0$ and
 not all $\dist_a(T,R),\dist_a(T',R),\dist_a(T\wedge T',R)$ are the same.
Then we have
$\min\{\dist_a(T,R),\dist_a(T',R)\}=\dist_a(T\wedge T',R)-1/2$
and $\max\{\dist_a(T,R),\dist_a(T',R)\}=\dist_a(T\wedge T',R)+1/2$.
Also $T_s\neq T'_s$ and thus $r[T\wedge T',T]=r[T\wedge T',T']=2$.
Then the three points $(-r[T\wedge T',T],\dist_a(T,R)),\ (0,\dist_a(T\wedge T',R))$
 and $(r[T\wedge T',T'],\dist_a(T',R))$ are collinear.
For (iii), we can assume $\dist_a(T\wedge T',R)>0$ similarly.
Then $\dist_a(T,R)=\dist_a(T\wedge T',R)+1/2$
and $\dist_a(T',R)\geq \dist_a(T\wedge T',R)-1/2$.
Moreover $r[T\wedge T',T]\leq r[T\wedge T',T']$.
Thus the point $(0,\dist_a(T\wedge T',R))$ is under or lying on the segment
connecting $(-r[T\wedge T',T],\dist_a(T,R))$ and $(r[T\wedge T',T'],\dist_a(T',R))$.

Suppose that $T\wedge T'=0$.
Then the $\wedge$-convexity inequality \cref{eq:anti} for (3) is immediate
from $\size(0)=0$ and the nonnegativity of $\size(\cdot)$.
For (4), we may assume $T\in \mathbb{W}_s\setminus\mathbb{C}$
 and $T'\in \mathbb{W}_{s'}\setminus\mathbb{C}$ for distinct $s,s'\in S$;
the other cases can be shown as the previous case.
Since $M$ is large, we may further assume
$T=[1/2,1/2]_s$ and $T'=[1/2,1/2]_{s'}$.
We may assume $\dist_a(0,R)>0$, and thus $R\notin \mathbb{C}$.
If $R\in \mathbb{W}_s\setminus\mathbb{C}$,
 then $\dist_a(T,R)+1/2=\dist_a(0,R)=\dist_a(T',R)-1/2$.
Thus three points $(-r[0,T],\dist_a(T,R))=(-2,\dist_a(T,R))$,
$(0,\dist_a(0,R))$, $(r[0,T'],\dist_a(T',R))=(2,\dist_a(T',R))$ are collinear.
If $R\in \mathbb{W}_t\setminus\mathbb{C}\ (t\in S\setminus\{s,s'\})$,
then $\dist_a(T,R)=\dist_a(0,R)+1/2=\dist_a(T',R)$.
Thus the point $(0,\dist_a(0,R))$ is under the segment
connecting $(-2,\dist_a(T,R))$ and $(2,\dist_a(T',R))$.

(5) By \cref{lem:Lconv_two},
it suffices to show that $(T,R)\mapsto \dist_a(T,R)$ satisfies
 the submodularity inequality \cref{eq:submod}
for any 2-bounded pair.
Let $(T,R),(T',R')\in \overline{\mathbb{S}}^2$ be a 2-bounded pair.
Then the pair $(T,T')$ (resp. $(R,R')$) is 2-bounded and $R=R'$ (resp. $T=T'$),
 or each pair $(T,T')$ and $(R,R')$ are joined with an edge in $\overline{\mathbb{S}}$.
In either case, there exists $s,t\in S$ such that
$T,R,T',R'$ are contained in $\mathbb{W}_s\cup \mathbb{C}\cup\mathbb{W}_t$.
Then by the definition,
\[
	\dist_a(T,R)=(-T_s-R_t-a)^++(-T_t-R_s-a)^+
\]
and so are $\dist_a(T',R'),\dist_a(T\wedge T',R\wedge R'),\dist_a(T\vee T',R\vee R')$.
Thus the submodularity inequality \cref{eq:submod} follows from \cref{lem:Lconv_grid}.
\end{proof}

We note that adding $B(\cdot)$ in \cref{prop:Lconv_term}~(4)(5) is not essential.
In fact, we can show the L-convexity of $\dist(\cdot,\cdot)$ itself on $\overline{\mathbb{S}}^2$.
We add $B(\cdot)$ just to simplify the proof.

\subsection{Sensitivity of DTB}
\label{sec:sensitivity}

Our goal is to show \cref{lem:sensitivity}.
We introduce a new concept of discrete convexity, called \emph{N-convexity}.
This concept is defined via \emph{normal paths}~\cite{Chepoi2000Graphs}
in a (general) median graph.
Instead of giving a general theory, we directly introduce N-convexity
on our special median graph $\overline{\mathbb{S}}^n$.
Let us first consider the grid graph $(\mathbb{Z}^*)^k$.
For $x,y\in (\mathbb{Z}^*)^k$, let define $x\to y\in (\mathbb{Z}^*)^k$ by
\begin{equation}
\label{eq:normal}
	x\to y:=x+\frac{1}{2}\sum_{i:x_i<y_i}\chi_i-\frac{1}{2}\sum_{i:x_i>y_i}\chi_i.
\end{equation}
Let $x\to^t y:=(x\to^{t-1} y)\to y$ for $t\geq 1$, where $x\to^0 y:=x$.
Observe that $\lVert (x\to^t y)-y\rVert_\infty
=\lVert (x\to^{t-1}y)-y\rVert_\infty-1/2$ for $t\leq 2\lVert x-y\rVert_\infty$.
Also $x\to^d y=y$, where $d:=2\lVert x-y\rVert_\infty$.
Let define $x\twoheadrightarrow y\in (\mathbb{Z}^*)^k$ by
\[
	x\twoheadrightarrow y:=y\to^{d-1}x=x+
	\frac{1}{2}\sum_{i:y_i-x_i=d/2}\chi_i-\frac{1}{2}\sum_{i:y_i-x_i=-d/2}\chi_i,
\]
where $d:=2\lVert x-y\rVert_\infty$.
A function $f:(\mathbb{Z}^*)^k\to \overline{\mathbb{R}}$ is called \emph{N-convex} if it satisfies
for any $x,y\in (\mathbb{Z}^*)^k$,
\begin{align}
    f(x)+f(y)&\geq f(x\to y) + f(y\to x),\label{cond:N1}\\
    f(x)+f(y)&\geq f(x\twoheadrightarrow y) + f(y\twoheadrightarrow x).\label{cond:N2}
\end{align}

As \cref{lem:Lconv_grid}, the next lemma can be routinely verified.

\begin{lemma}[\cite{Hirai2020cost}]
\label{lem:Nconv_grid}
\begin{enumerate}
\renewcommand{\labelenumi}{\textup{(\arabic{enumi})}}
\item $(\mathbb{Z}^*)^k\ni x\mapsto a^\top x+b$ is N-convex for $a\in \mathbb{R}^k,\ b\in\mathbb{R}$.
\item $(\mathbb{Z}^*)^k\ni x\mapsto (x_i+x_j-a)^+$ is N-convex for $i,j\in\{1,\dotsc,k\},\ a\in\mathbb{Z}$.
\item The N-convexity does not change by an individual sign inversion of a variable,
i.e., if a function $f:(\mathbb{Z}^*)^k\to \mathbb{R}$ is N-convex,
 then $x\mapsto f(x_1,\dotsc,-x_i,\dotsc,x_k)$ is also N-convex.
\end{enumerate}
\end{lemma}

Let us return to our graph $\overline{\mathbb{S}}^n$.
Let $T,T'\in\overline{\mathbb{S}}$.
Since $\overline{\mathbb{S}}$ is a subset of $(\mathbb{Z}^*)^k$,
$T\to T'\in\overline{\mathbb{S}}$ is defined via \cref{eq:normal},
except for (1) $(T,T')\in (\mathbb{W}_s\setminus\mathbb{C})\times (\mathbb{W}_t\setminus\mathbb{C})\ (s,t\in S,\ s\neq t)$
or (2) $(T,T')\in (\mathbb{C}\setminus\mathbb{W}_s)\times(\mathbb{W}_s\setminus\mathbb{C})\ (s\in S)$.
In these cases, it may happen that \cref{eq:normal} is outside of $\overline{\mathbb{S}}$.
Instead, let define for (1),
\[
	T\to T':=\begin{dcases*}
		T+\frac{1}{2}\chi_{S-s}-\frac{1}{2}\chi_s & if $T_s>0$,\\
		T+\frac{1}{2}\chi_{S-s} & if $T_s=0$,
	\end{dcases*}
\]
and for (2),
\[
	T\to T':=T+\frac{1}{2}\sum_{t:T_t<T'_t}\chi_t-\frac{1}{2}\sum_{t:0<T_t>T'_t}\chi_t.
\]
Then we see ${T\to T'}\in \overline{\mathbb{S}}$.
Let ${T\to^t T'}:={({T\to^{t-1} T'})\to T'}$ for $t\geq 1$, where ${T\to^0 T'}:=T$.
We can see the desired property still holds:
$\lVert ({T\to^t T'})-T'\rVert=\lVert ({T\to^{t-1} T'})-T'\rVert-1/2$ for $t\leq d$
and ${T\to^d T'}=T'$, where $d:=d^\Delta(T,T')\,
(=2\lVert T-T'\rVert\text{ if $T,T'\in \mathbb{S}$})$.
The $\Delta$-path $(T={T\to^0 T'},{T\to^1 T'},\dotsc,{T\to^d T'}=T')$
is called the \emph{normal path} from $T$ to $T'$.
Let ${T\twoheadrightarrow T'}:={T'\to^{d-1} T}\in\overline{\mathbb{S}}$.
For $p,q\in \overline{\mathbb{S}}^l$, we define ${p\to q}$
 by ${(p\to q)_i}:={p_i\to q_i}$,
and let ${p\to^t q}:={({p\to^{t-1} q})\to q}$ for $t\geq 1$,
where ${p\to^0 q}:=p$,
 and $p\twoheadrightarrow q:={p\to^{d-1} q}$,
where $d=d^\Delta(p,q)$.
A function $f:\overline{\mathbb{S}}^l\to\overline{\mathbb{R}}$
 (in particular, $f:\mathbb{G}\to\overline{\mathbb{R}}$) is called \emph{N-convex}
if it satisfies \cref{cond:N1,cond:N2}
 for any $p,q\in \overline{\mathbb{S}}^l$.

Let define $h:\mathbb{G}\to \overline{\mathbb{R}}$ by
\begin{equation}
\label{eq:hN}
	h(p):=\begin{dcases*}
		-\sum_{s\in S}r_s\dist(0,p_s)+\sum_{i\in V\setminus S}c_i\size(p_i)
		+\sum_{ij\in E}u_{ij}\dist_{a_{ij}}(p_i,p_j) & if $p$ is a potential,\\
		\infty & otherwise.
	\end{dcases*}
\end{equation}
That is, $h$ is obtained from $\tilde{h}$
by taking $M\to \infty$.
Since \cref{eq:hN} is consistent with \cref{eq:horig} on $\mathbb{G}\cap\mathbb{S}^n$,
we use $h$ for denoting \cref{eq:hN}.

\begin{proposition}
\label{prop:Nconv_h}
$h$ is N-convex on $\mathbb{G}$.
\end{proposition}

First we prove \cref{lem:sensitivity} using \cref{prop:Nconv_h}.

\begin{proof}[Proof of \cref{lem:sensitivity}]
Let $e=ij$.
Recall that $a':=a-\chi_{ij}$.
Then we observe that
\begin{equation}
\label{eq:h_diff}
	h_{a'}(q)=\begin{dcases*}
		h_a(q)+u_{ij} & if $\dist(q_i,q_j)\geq a_{ij}$,\\
		h_a(q)+u_{ij}/2 & if $\dist(q_i,q_j)=a_{ij}-1/2$,\\
		h_a(q) & if $\dist(q_i,q_j)\leq a_{ij}-1$.
	\end{dcases*}
\end{equation}
Take $q\in \opt(h_{a'})$ having the minimum $\Delta$-distance from $p$.
Suppose that $q\neq p$.
Then $h_{a'}(q)<h_{a'}(p)$ and also $h_a(p)\leq h_a(q)$.
Let $(p=p^0,p^1,\dotsc,p^{\ell}=q)$ be the normal path from $p$ to $q$.
We show
\begin{gather}
\label{eq:h1_chain}
	h_a(p)\leq h_a(p^1)\leq h_a(p^2)\leq \dotsb\leq h_a(p^{\ell-1})\leq h_a(q),\\
\label{eq:h2_chain}
	h_{a'}(p)> h_{a'}(p^1)> h_{a'}(p^2)> \dotsb > h_{a'}(p^{\ell-1}) > h_{a'}(q).
\end{gather}
By \cref{prop:Nconv_h}, we have
$h_a(p)+h_a(p^{t+1})\geq h_a(p\twoheadrightarrow p^{t+1})+h_a(p^t)$.
It follows from $h_a(p)\leq h_a(p\twoheadrightarrow p^{t+1})$
that $h_a(p^{t+1})\geq h_a(p^t)$.
Similarly, we have
$h_{a'}(p^{t})+h_{a'}(q)\geq h_{a'}(p^{t+1})+h_{a'}(q\to p^t)$ by \cref{prop:Nconv_h}.
Since $d^\Delta(p,q\to p^t)<d^\Delta(p,q)$, it holds $h_{a'}(q)<h_{a'}(q\to p^t)$.
Then $h_{a'}(p^t)>h_{a'}(p^{t+1})$ follows.
By \cref{eq:h1_chain,eq:h2_chain}, we see
\begin{equation}
\label{eq:hdiff_chain}
	h_{a'}(p)-h_a(p)>h_{a'}(p^1)-h_a(p^1)>\dotsb
	 > h_{a'}(p^{\ell-1})-h_a(p^{\ell-1}) > h_{a'}(q)-h_a(q).
\end{equation}
Now $\ell\leq 2$ follows from \cref{eq:h_diff,eq:hdiff_chain}.
\end{proof}

This proof method is also used in \cite{Hirai2020cost}.

We prove \cref{prop:Nconv_h}.
It is straightforward that the N-convexity is preserved
under the addition and the multiplication of a nonnegative real number.
Moreover, if $f$ is N-convex on $\overline{\mathbb{S}}^l$,
then $\tilde{f}:\overline{\mathbb{S}}^{l+1}\to \mathbb{R}$ defined by
\[
	\tilde{f}(Q,T):=f(Q)\quad (Q\in\overline{\mathbb{S}}^l,
		\ T\in\overline{\mathbb{S}})
\]
is also N-convex.
To see this for \cref{cond:N2},
notice that $(\tilde{f}((Q,T)\twoheadrightarrow (Q',T')),\tilde{f}((Q',T')\twoheadrightarrow (Q,T)))$
is $(f(Q),f(Q'))$ or $(f(Q\twoheadrightarrow Q'),f(Q'\twoheadrightarrow Q))$.
Then we only need to show the following.

\begin{proposition}
\label{prop:Nconv_term}
The functions
\textup{(1)} $-\dist(0,\cdot)$ on $\mathbb{W}_s$,
\textup{(2)} $\size(\cdot)+\infty\cdot B(\cdot)$ on $\overline{\mathbb{S}}$,
and
\begin{align*}
	\text{\textup{(3)}}&\quad \overline{\mathbb{S}}^2\ni (T,R)\mapsto \dist_a(T,R)+\infty\cdot(B(T)+B(R))\quad (0<a\in \mathbb{Z}_+)
\end{align*}
are all N-convex.
\end{proposition}

Here $\infty\cdot B(T)$ takes value $0$ if $B(T)=0$ and $\infty$ otherwise.
This means we may only consider the points in $\mathbb{S}$.

\begin{proof}
(1) is obvious from \cref{lem:Nconv_grid}~(1).
(2) Let $T,Q\in \mathbb{S}$ and
let $T':=T\to Q,\ Q':=Q\to T$
(or $T':=T\twoheadrightarrow Q,\ Q':=Q\twoheadrightarrow T$);
our argument can be applied to both cases.
We can assume $T\in \mathbb{W}_s\setminus\mathbb{C}\ (s\in S)$
 and $Q\in \mathbb{C}\setminus\bigcup_{s\in S} \mathbb{W}_s$;
 other cases can be reduced to \cref{lem:Nconv_grid}~(1) and (2).
Let $t\in S-s$.
Then $\size(T)=T(s)+T(t)$.
We have $T(s)-T'(s)=Q'(s)-Q(s)$.
Also $T(t)\leq T'(t)$ and $Q'(t')\leq Q(t')$ for any $t'\in S-s$.
If $T(t)<T'(t)\ (=T(t)+1/2)$,
then at least one of $t'\in S-s$ satisfies $Q'(t')=Q(t')-1/2$.
Thus
\[
	\size(T)+\size(Q)-\size(T')-\size(Q')
	\geq T(t)-T'(t)+Q(t')-Q'(t')=0.
\]

(3) Let $(T_1,T_2),(Q_1,Q_2)\in \mathbb{S}^2$
and let $(T'_1,T'_2):=(T_1,T_2)\to(Q_1,Q_2),\ 
(Q'_1,Q'_2):=(Q_1,Q_2)\to(T_1,T_2)$
(or $\to$ replaced by $\twoheadrightarrow$).
To show \cref{cond:N1,cond:N2},
we may assume that $\dist_a(T_1,T_2)<\dist_a(T'_1,T'_2)\ (>0)$.
Since $\dist(T'_1,T'_2)-\dist(T_1,T_2)\in\{1/2,1\}$,
we have $\dist(T_1,T_2)\geq \dist(T'_1,T'_2)-1>a-1\geq 0$.
Then at least one of $T_1,T_2$ does not belong to $\mathbb{C}$,
say, $T_1\in \mathbb{W}_s\setminus\mathbb{C}$.
If $T_2$ also belongs to $\mathbb{W}_s$,
then we can further assume that $\dist(0,T_1)>\dist(0,T_2)$.
Hence $\dist(T_1,T_2)=-T_1(t)-T_2(s)\ (\forall t\in S-s)$.
We also see $T'_1\in \mathbb{W}_s$ and
$\dist(T'_1,T'_2)=-T'_1(t)-T'_2(s)$.

Suppose that $Q_1\in \mathbb{W}_s$ (and thus $Q'_1\in \mathbb{W}_s$).
Then we can assume there exists $t\in S-s$ such that
 $T_2,Q_2$ (and $T'_2,Q'_2$) are belong to
$\mathbb{W}_s\cup \mathbb{C}\cup \mathbb{W}_t$.
Indeed, if $T_2\in \mathbb{W}_t\setminus\mathbb{C}$
 and $Q_2\in \mathbb{W}_{t'}\setminus\mathbb{C}$
for distinct $t,t'\in S-s$,
then this can be reduced to the next case by interchanging indices 1 and 2.
We see $\dist(Q_1,Q_2)=-Q_1(t)-Q_2(s)$ and $\dist(Q'_1,Q'_2)=-Q'_1(t)-Q'_2(s)$.
We consider the projection $\sigma$ onto $(\mathbb{Z}^*)^{\{s,t\}}$
as in \cref{prop:Lconv_term}~(5).
For $(T_1,T_2),(Q_1,Q_2)$, the projection $\sigma$ and
 $\to$ ($\leftarrow,\twoheadrightarrow,\twoheadleftarrow$) commute, i.e.,
$(\sigma(T_1),\sigma(T_2))\to(\sigma(Q_1),\sigma(Q_2))
=(\sigma(T'_1),\sigma(T'_2))$ and so on
except for the case of $\to$, $T_2(s)=0$, $Q_2(s)<0$
but there exists $t'\in S\setminus\{s,t\}$ with $T_2(t')>0$.
In the commuting cases, \cref{cond:N1,cond:N2} follow from \cref{lem:Nconv_grid}~(2).
For the exceptional case,
$T_2(s)=T'_2(s)$ but $\sigma(T_2)(s)>(\sigma(T_2)\to \sigma(Q_2))(s)$.
Then $\dist(T'_1,T'_2)<\dist((\sigma(T_1),\sigma(T_2))\to(\sigma(Q_1),\sigma(Q_2)))$ holds,
and hence \cref{cond:N1} holds.

Suppose that $Q_1\notin \mathbb{W}_s$.
This happens only if $(T'_1,T'_2):=(T_1,T_2)\twoheadrightarrow (Q_1,Q_2)$
(and $(Q'_1,Q'_2):=(Q_1,Q_2)\twoheadrightarrow (T_1,T_2)$)
and $-T_1(t)=-T'_1(t)\geq T_2(s)=T'_2(s)+1/2\ (\forall t\in S)$.
We take $t\in S-s$ with maximum $Q_1(t)>0$.
By the definition of ``$\twoheadrightarrow$'',
we have
\begin{equation}
\label{eq:increase}
	T_2(s)-Q_2(s)>Q_1(t)-T_1(t)
\end{equation}
and thus $Q_2(s)< T_2(s)+T_1(t)-Q_1(t)<T_2(s)+T_1(t)\leq 0$.
Hence $Q_2\in \mathbb{W}_{t'}\setminus\mathbb{C}$
 for some $t'\in S-s$ (possibly $t'=t$).
By \cref{eq:increase},
$\dist(Q_1,Q_2)=-Q_1(t')-Q_2(s)\geq -Q_1(t)-Q_2(s)>-T_1(t)-T_2(s)=\dist(T_1,T_2)\geq a$.
Similarly $Q'_2\in \mathbb{W}_{t'}$ and
$\dist(Q'_1,Q'_2)=-Q'_1(t')-Q'_2(s)\geq a$.
Also by the definition of ``$\twoheadrightarrow$'',
we have $Q_1(t')-Q'_1(t')\in \{0,-1/2\}$ and $Q'_2(s)=Q_2(s)+1/2$.
Hence $\dist_a(Q'_1,Q'_2) \leq \dist_a(Q_1,Q_2) -1/2$, implying \cref{cond:N2}.
\end{proof}

\section*{Acknowledgements}
The first author was supported by JSPS KAKENHI Grant Number JP17K00029
and JST PRESTO Grant Number JPMJPR192A, Japan.


\bibliography{ntb}

\end{document}